%% file: draft-2026-10.tex
\documentclass[12pt]{article}
\usepackage{amssymb}
\usepackage{amsmath}
\usepackage{amsfonts}
\usepackage{geometry,enumerate,bbm}
\usepackage{natbib}
\usepackage{epsfig}
\usepackage{float}
\usepackage[onehalfspacing]{setspace}
\usepackage{graphicx}
\usepackage{subcaption}

\usepackage{lscape}
\usepackage[flushleft]{threeparttable}
\usepackage{ulem}
\usepackage{scrextend}
\usepackage{xcolor}

\usepackage{booktabs}
\usepackage{multicol}

\definecolor{dgreen}{rgb}{0,0.5,0}

\RequirePackage[colorlinks,citecolor=blue,urlcolor=blue,linkcolor=blue]{hyperref}
\newtheorem{assumption}{Assumption}
\newtheorem{definition}{Definition}
\newtheorem{theorem}{Theorem}
\newtheorem{conjecture}{Conjecture}

\newtheorem{lemma}{Lemma}
\newtheorem{proposition}{Proposition}
\newtheorem{remark}{Remark}
\newenvironment{proof}[1][Proof]{\noindent \textbf{#1.} }{\  \rule{0.5em}{0.5em}}
\newcommand{\argmax}{\operatorname*{argmax}}

\newcommand{\norm}[1]{\lVert #1 \rVert}
\newcommand{\abs}[1]{|#1|}

\usepackage{physics}

\begin{document}

\title{{\bf Dynamic Panel Multinomial Logit Models}
\thanks{
This research was
supported by the Gregory C. Chow Econometric Research Program at Princeton
University and by  the Danish National Research Foundation grant number DNRF186 to the Aarhus Center for Econometrics (ACE). 
We are grateful to seminar participants at numerous institutions for constructive comments.}
}
\author{\setcounter{footnote}{2}
Kevin Dano%
\thanks{Princeton University, \texttt{kdano@princeton.edu} }
\and
Bo E. Honor{\'e}%
\thanks{Princeton University and Aarhus Center for Econometrics (ACE) at the University of Aarhus, \texttt{honore@princeton.edu} }
\and
 Martin Weidner%
\thanks{
University College London, \texttt{m.weidner@ucl.ac.uk} } }
\date{October 2026}

\maketitle
\thispagestyle{empty}
\setcounter{page}{0}

\begin{abstract}
{
\noindent This paper provides moment conditions for dynamic multinomial logit panel data models with individual specific fixed effects. The moment conditions do not depend on the fixed effects and %
thus form the %
basis for generalized method of moments estimation. For a leading case, we provide sufficient conditions under which the moment conditions identify the common parameters, and we illustrate their use %
in an empirical application to household soft drink purchases.
The construction of moment conditions is based on transformed models obtained by permuting outcomes across time periods. For any cyclic permutation, the ratio of the transformed model probability to the true model probability is free of the fixed effects, and comparing two such ratios yields %
restrictions on the common parameters of interest. %
This construction is not specific to the multinomial logit. It applies to any dynamic panel data model whose static counterpart admits a symmetric sufficient statistic for the fixed effects. This unifies our main result with analogous results for other models.
}
\end{abstract}

\newpage

\section{Introduction}

Controlling for individual heterogeneity without imposing strong assumptions on
its relationship to observed covariates is a central challenge in the
econometric analysis of nonlinear panel data models. Suppose the individual-specific effects are treated as parameters and estimated jointly with the common parameters of interest. The resulting estimator is then generally inconsistent as the number of cross-sectional units grows with the number of time periods held fixed. This is the classical incidental parameter problem \citep{neyman1948consistent}. This paper studies the dynamic multinomial logit model for panel data with fixed effects. Our main contribution is to construct moment conditions that depend on the common parameters but not on the individual-specific effects. We obtain them from a transformed model that permutes outcomes across time periods, and that cancels the fixed effects through the structure of the multinomial logit likelihood.

One approach to panel data discrete choice models with fixed effects is to look for sufficient statistics for the fixed effects. It dates back to \citet{rasch1960studies}, who observed that in a binary logit model with strictly exogenous regressors the sum of an individual's outcomes is a sufficient statistic for the fixed effect, so that conditional maximum likelihood can be used. \citet{chamberlain_analysis_1980} extended this to the static multinomial logit, where the vector of choice frequencies is sufficient for the fixed effects, and \citet{magnac2000subsidised} obtained an analogous result when the only regressors are lagged dependent variables. Unfortunately, conditional likelihood methods generally break down in the empirically most relevant case in which both lagged dependent variables and strictly exogenous regressors are present. \citet{honore2000panel} handled this joint case for binary and multinomial logit by additionally conditioning on covariates being equal across certain periods. This leads to an estimator whose rate of convergence is slower than $\sqrt{n}$. \citet{Kyriazidou2001} combined conditional likelihood with kernel smoothing in dynamic panel sample selection models. \citet{KhanOuyangTamer2021} generalized
\citet{honore2000panel} to semiparametric multinomial response models that drop
the logit assumption and relax the independence of irrelevant alternatives,
establishing point identification but again at a rate slower than $\sqrt{n}$
that deteriorates with the number of alternatives. 

Since there are relatively few models for which sufficient statistics can be found, an alternative approach is to try to construct moment conditions that depend on the common parameters, but not on the individual-specific fixed effects. In the linear case this is what \citet{ArellanoBond1991} do, when they eliminate the fixed effects by first-differencing and instrument with lagged levels. Papers by \citet{Honore1993} and \citet{Hu2002} are specific examples of this for panel Tobit models, where the conditions are obtained from a symmetry that is restored through trimming, and the argument is tailored to the censored-regression structure rather than being generic. For multiplicative models, \citet{Chamberlain1992} and \citet{Wooldridge1997} derived moment conditions which are valid without strict exogeneity, implying that lagged dependent variables can be accommodated in count and other nonnegative outcome models. \citet{bonhomme2012functional} developed a general approach for obtaining moment conditions via functional differencing, which essentially differences out the unknown distribution of the individual effects in the likelihood. While general, this characterization does not by itself deliver explicit moment functions, and concrete constructions have so far proceeded model by model, exploiting the specific structure of each likelihood.
Functional differencing has been operationalized for the dynamic binary logit
model by \citet{kitazawa2022transformations} and \citet{honore2025moment},
extended to the dynamic ordered logit by \citet{HonoreMurisWeidner2025}, and
given a general transition-probability treatment by \citet{dano2023transition}.
The present paper extends this program to the dynamic multinomial logit.\footnote{ \cite{AlSadoonLiPesaranYangZhao2026} also consider moment conditions in a dynamic multinomial fixed effects model, but their model is outside the multinomial logit framework and it places restrictions on the support of the fixed effects that depend on the explanatory variables and the model parameters.}

Multinomial outcomes are common in applied work with panel data. They arise in choices among employment states \citep{keane1997career}, consumer products with repeat purchasing and switching costs  \citep{handel2013adverse, dubehitschrossi2010}, health insurance plans \citep{pakes2021unobserved}, urban consumption locations \citep{davis2019segregated}, and migration destinations \citep{blumenstock2025migration}. All of these fit naturally into the multinomial framework. Because switching costs and habit formation are prominent in such settings, it is important to allow lagged choices to enter as regressors and to control for permanent unobserved heterogeneity at the same time \citep{Heckman1981}. 

The construction of our moment conditions is based on a transformed model in which the outcome generated in each period is reassigned according to a permutation $\sigma$ of the time indices. For any $\sigma$, the ratio of the transformed-model probability to the true-model probability does not depend on the fixed effects. This follows from the logistic structure, in which the fixed effects enter the choice probabilities via alternative-specific constants whose total contribution depends on the unordered collection of outcomes only --- implying invariance to reshuffling outcomes across time. Intuitively, reshuffling outcomes across
time leaves the total contribution of the individual effects unchanged, so they
can be differenced out through likelihood ratios. When $\sigma$ is a cyclic
permutation the transformed model is a proper probability distribution, and
comparing two such models yields moment conditions of the form required for
fixed effects-free estimation; the binary logit conditions of
\citet{honore2025moment} emerge as the special case with two alternatives. The construction is not specific to the multinomial logit. All that is needed is that the static version of the model admits a sufficient statistic for the fixed effects that is symmetric in the time indices, i.e. a statistic that depends on the outcomes only through their unordered collection. This unifies our results with analogous
ones for other models and shows that the approach applies to any dynamic panel
model built on a static model that admits conditional likelihood estimation via
a symmetric sufficient statistic.

The number of linearly independent moment conditions for
given $(Q,T)$ grows rapidly in both the number of alternatives $Q$ and the
number of time periods $T$, and we  conjecture an explicit formula for this count.
We prove that this formula is correct for $Q=2$, and that it is a
lower bound for $T=3$ and any $Q$. For the other combinations of $(Q,T)$ that we
report we verify it numerically. At
least three time periods after the initial condition are required, as in the
binary case. We give sufficient conditions under which these moments identify
the common parameters, the regression and state-dependence coefficients, and
propose a generalized method of moments (GMM) estimator. Because the number of
available moments grows quickly, only a small subset is typically used in
estimation. An empirical application to household soft drink purchases
illustrates the importance of controlling for unobserved heterogeneity when
estimating state dependence in multinomial choice.

 The paper is organized as follows. Section~\ref{sec:model} introduces the model. Section~\ref{sec:moments} presents the moment conditions and establishes their validity. Section~\ref{sec:moregen} generalizes the construction to other dynamic models whose static versions admit a symmetric sufficient statistic for the fixed effects. Section~\ref{sec:identification} analyzes identification. Section~\ref{sec:empirical} contains the empirical application. Section~\ref{sec:conclusion} concludes.

\section{Multinomial logit model}
\label{sec:model}

Panel settings in which agents repeatedly select from a finite set of alternatives motivate our framework. Examples of such alternatives are brands of a consumer good, transportation modes, employment states, or health insurance plans. Our analysis focuses on two features of these settings: agents differ in their persistent preferences over alternatives (a variation not observed by the econometrician), and past choices can influence current utility (which gives rise to state dependence). Both features produce persistence in the observed choices,
so distinguishing them is difficult on the basis of cross-sectional data alone. This is why we use panel data where the same agents are observed over multiple time periods.

We consider panel data with cross-sectional units $i=1,\ldots,n$ and time periods $t=0,\ldots,T$.
For each pair $(i,t)$, we observe the discrete outcome $Y_{it} \in \{1,2,\ldots,Q\}$, where  $Q \in \{2,3,4,\ldots\}$ denotes the number of options.
We also observe strictly exogenous regressors $X_{it} \in \mathbb{R}^K$ for $t \geq 1$.
The observed outcome at $t=0$ serves as the initial condition and is conditioned on throughout. For $t \geq 1$, the outcomes are generated by the model
\begin{align}
Y_{it}=\argmax_{q \in \left\{ 1,2,\ldots,Q\right\} }U_{qit} ,
    \label{ModelChoice}
\end{align}
with latent utilities
\begin{align}
    U_{qit} &= z_t(q,Y_{i,t-1},X_i,\theta) +  A_{qi} + \varepsilon _{qit} .
    \label{ModelUtility}
\end{align}
Here, $A_i=(A_{1i},\ldots,A_{Qi})$ are individual-specific fixed effects, and
$X_i=(X_{i1},X_{i2},\ldots,X_{iT})$ are all the strictly exogenous covariates observed for unit $i$. The errors
$\varepsilon _{qit}$ are independent of $(Y_{i0},X_i,A_i)$, and are assumed to be independent and identically distributed (i.i.d.)\ over $t$ and $q$ with a type I extreme value distribution.

The vector $A_i$ collects alternative-specific constants for unit $i$, with $A_{qi}$ capturing the taste for option $q$. Since only differences between the $A_{qi}$ enter the choice probabilities, one can normalize $A_{1i} = 0$. 
The type I extreme value assumption implies that the choice probabilities in \eqref{CondProb} below take the form of the familiar multinomial logit model. The independence of irrelevant alternatives across alternatives within each period is therefore the key substantive restriction. Conditioning on $Y_{i0}$ throughout avoids the need to specify a model for the initial outcome. This is convenient, because $Y_{i0}$ is typically correlated with the fixed effects in an unknown way.

The parameter $\theta$ is unknown, but the function $z_t(\cdot,\cdot,\cdot,\theta)$ is known up to $\theta$.
The standard linear specification reads
  \begin{align}
    z_t(q,Y_{i,t-1},X_i,\theta) =  \sum_{r=1}^Q  \, \gamma_{qr} \,  \mathbbm{1}\left\{Y_{i,t-1} = r \right\} + X_{it}' \, \beta_q ,
    \label{LinearSpec}
 \end{align}
with $\theta=(\gamma,\beta)$. %
Our primary interest is in the state dependence coefficients $\gamma_{qr}$ and the regression coefficients $\beta_q$. The diagonal entries $\gamma_{qq}$ measure how much the utility of alternative $q$ rises when $q$ was also chosen in the previous period. In other words, $\gamma_{qq}$ captures the degree of habit formation or switching costs for option $q$. The off-diagonal entries $\gamma_{qr}$ with $r \neq q$ capture additional forms of persistence, including asymmetries between related alternatives. The central difficulty in estimating the persistence coefficients is that one would like to separate genuine state dependence (in the form of the $\gamma$'s) from spurious persistence induced by unobserved heterogeneity. See, for example, \citet{Heckman1981}. We address this problem by constructing moment conditions that are free of $A_i$ but informative about the $\gamma$'s.

The coefficients in the model above are identified only up to normalizations. Three types of shifts leave the choice probabilities unchanged: (i) shifting $\gamma_{qr}$ by a constant that depends on $r$ but not on $q$, because such a shift appears in the numerator and in every term of the denominator of the logit and therefore cancels; (ii) shifting $\gamma_{qr}$ by a constant that depends on $q$ but not on $r$, because $\sum_{r=1}^Q \mathbbm{1}\{Y_{i,t-1} = r\} = 1$ and such a shift is therefore equivalent to adding a constant to the utility of alternative $q$, which is absorbed into the fixed effect $A_{qi}$; and (iii) shifting $\beta_q$ by an additive vector in $\mathbb{R}^K$ that does not depend on $q$, which cancels between the numerator and denominator of the logit. To remove these sources of non-identification we impose the normalizations $\gamma_{1r} = 0$ for all $r \in \{1, \ldots, Q\}$, $\gamma_{q1} = 0$ for all $q \in \{1, \ldots, Q\}$, and $\beta_1 = 0$. The first of these takes alternative $1$ as the reference alternative, as in \citet{magnac2000subsidised}, and together with $A_{1i} = 0$ handles the non-identification of type (i). The second takes lagged outcome $1$ as the reference lagged state and handles the non-identification of type (ii). The third normalization handles type (iii).

Of course, the baseline linear specification in \eqref{LinearSpec} is only one possible parametrization of $z_t(\cdot,\cdot,\cdot,\theta)$.  The framework allows for more flexible choices. For example, the parameters $\gamma_{qr}$ and $\beta_q$ in \eqref{LinearSpec} could depend on $t$, or $Y_{i,t-1}$ and $X_{it}$ could be interacted. Conversely, in some applications a more parsimonious specification may be desirable. In our empirical application, we consider  
\begin{align} \label{EmpiricalSpec}
z_{t}(q,Y_{i,t-1},X_i,\theta)=\gamma_{q} \mathbbm{1}\left\{Y_{i,t-1}=q\right\} + X_{it,q}\beta   
\end{align}
where the $Q$ coefficients $\gamma_{q}$ reward habit formation, while the effect $\beta$ of the alternative-specific covariate $X_{it,q}$ (i.e price of good $q$) %
is restricted to be constant across alternatives.\footnote{Specification \eqref{EmpiricalSpec} can be viewed as a special case of \eqref{LinearSpec} where the cross-state dependence parameters satisfy $\gamma_{qr}=0$ for $q\neq r$, and the alternative-specific slopes are restricted to $\beta_q=\beta\,\iota_q$, where $\iota_q$ is the $q$th canonical basis vector. Hence $X_{it}'\beta_q=X_{it,q}\beta$, explaining the notational change between \eqref{LinearSpec} and \eqref{EmpiricalSpec}.} For $Q \geq 3$, this specification does not require the normalizations imposed above,  since the parameter vector is
$\theta=(\beta,\gamma_1,\ldots,\gamma_Q)$. We discuss this model further in Sections \ref{sec:identification} and \ref{sec:empirical}.

Let $Y_i=(Y_{i1},Y_{i2},\ldots,Y_{iT})$.
We assume cross-sectional sampling, that is, $(Y_i,Y_{i,0},X_i,A_i)$ are i.i.d.\ distributed across $i$.
For convenience, we drop the subscript $i$ from the notation when no confusion can arise.
Then, the model assumptions imply that
  \begin{align}
P\left(  Y_{t}=q  \, \Big| \,  Y^{t-1}, X, A \right) =\frac{\exp \left(z_t(q,Y_{t-1},X,\theta) +A_q\right) }{\sum_{r=1}^{Q}
\exp\left( z_t(r,Y_{t-1},X,\theta)  +A_r \right) } ,
    \label{CondProb}
 \end{align}
where $Y^{t-1}=(Y_{t-1},Y_{t-2},\ldots,Y_0)$ denotes past outcomes.
It is convenient to introduce the shorthand notation\footnote{%
Regarding notation, uppercase $X$ and $A$ denote random vectors, whereas lowercase $x$ and $\alpha$ denote their possible non-random realizations. The object $p_t(q \mid y_{t-1}, x, \theta, \alpha)$ is defined as an ordinary function, so its arguments are non-random. By contrast, stochastic operators such as $P$ and $\mathbb{E}$ act on random quantities and therefore involve $X$ and $A$.}
\begin{align}
    p_t(q \mid y_{t-1}, x, \theta, \alpha) := \frac{\exp \left(z_t(q,y_{t-1},x,\theta) +\alpha_q\right) }{\sum_{r=1}^{Q}
\exp\left( z_t(r,y_{t-1},x,\theta)  +\alpha_r \right) }
\label{DefPt}
\end{align}
for the conditional probability of choosing alternative $q$ in period $t$, given that the previous choice was $y_{t-1}$.
With this notation, the joint probability of observing a sequence of outcomes $y=(y_1,\ldots,y_T)$ given the initial condition and fixed effects is
  \begin{align}
p(y \mid y_0, x, \theta, \alpha) := P\left(  Y=y  \, \Big| \,  Y_0=y_0 , X=x, A=\alpha \right) =
 \prod_{t=1}^T p_t(y_t \mid y_{t-1}, x, \theta, \alpha) .
      \label{JointProb}
 \end{align}
The true model parameters are denoted by $\theta_0$. Throughout, we assume that the model in \eqref{JointProb} is correctly specified and all probabilistic statements below are for the model distribution generated under $\theta_0$.

Our goal is to estimate the common parameters $\theta_0$ from the observed data as $n \rightarrow \infty$ with $T$ fixed. The individual-specific fixed effects $A_i$ are treated as nuisance parameters. Their distribution is left unspecified, and they are allowed to be arbitrarily correlated with the regressors $X_i$ and the initial conditions $Y_{i0}$. In this setting, treating the individual-specific fixed effects, $A_i$, as parameters to be estimated will lead  to the incidental parameter problem \citep{neyman1948consistent}.

\section{Moment conditions free of fixed effects}
\label{sec:moments}

\subsection{Moment condition approach}

A classic approach to the incidental parameter problem in parametric models is to look for sufficient statistics for the fixed effects. Conditional on these sufficient statistics, the likelihood does not depend on the fixed effects and can be used for estimation. This approach was pioneered by \cite{rasch1960studies} and \cite{andersen1970asymptotic}. For the static multinomial logit model without regressors, \cite{chamberlain_analysis_1980} showed that the vector of choice frequencies serves as a sufficient statistic. For dynamic models, \cite{magnac2000subsidised} and \cite{aguirregabiriaSufficientStatisticsUnobserved2021} have developed sufficient statistic approaches. However, these methods do not generally extend to dynamic multinomial logit models with strictly exogenous regressors.

An alternative approach is to construct moment conditions that depend on the common parameters but not on the individual-specific fixed effects. \cite{bonhomme2012functional} developed a general framework for obtaining such moment conditions via ``functional differencing.'' This approach has been applied to dynamic binary logit models by \cite{kitazawa2022transformations} and \cite{honore2025moment}, and to dynamic ordered logit models by \cite{HonoreMurisWeidner2025}. \cite{dano2023transition} provides a general treatment using transition probabilities. The present paper extends this approach to the dynamic multinomial logit model. 

Formally, we seek moment functions $m_{y_0}(y,x,\theta) \in \mathbb{R}$ such that
\begin{align}
\mathbb{E}\left[ m_{Y_0}(Y,X,\theta_0)
   \, \Big| \, Y_0,X,A \right] = 0
   \label{MomentCondition}
\end{align}
for all possible values of the fixed effect $A$. Conditioning on the entire covariate path $X$ here uses the strict exogeneity assumption from Section~\ref{sec:model}, which rules out feedback from past outcomes to future covariates.
For given values of $Y_0=y_0$, $X=x$, $\theta$, and $A=\alpha$, the distribution of $Y=(Y_1,\ldots,Y_T) \in \{1,\ldots,Q\}^T$ is given by \eqref{JointProb}. The conditional moment condition \eqref{MomentCondition} can therefore be written as
\begin{align}
    \sum_{y \in \{1,\ldots,Q\}^T} m_{y_0}(y,x,\theta) \, p(y \mid y_0, x, \theta, \alpha) = 0 ,
     \qquad \text{for all } \alpha \in \mathbb{R}^Q.
     \label{MomentConditionSum}
\end{align}
This is a system of linear equations in the $Q^T$ unknowns $m_{y_0}(y,x,\theta)$, where $y$ ranges over $\{1,\ldots,Q\}^T$. The system must hold for all values of $\alpha$, which imposes infinitely many constraints. The existence of solutions to this infinite-dimensional system is not guaranteed for arbitrary models, and the distributional assumption on the errors is important for finding valid moment conditions.  This linear-algebra view of the problem is the starting point of the functional differencing approach of \citet{bonhomme2012functional}: valid moment functions are those that lie in the joint null space of the constraints generated by all $\alpha \in \mathbb{R}^Q$.

If we can establish the conditional moment condition \eqref{MomentCondition}, then by the law of iterated expectations, we also have the unconditional moment conditions
\begin{align}
\mathbb{E} \left[ h(Y_0,X,\theta_0) \,  m_{Y_0}(Y,X,\theta_0)   \right] &= 0
   \label{MomentUnconditional}
\end{align}
for any function $h$ such that the expectation is well-defined. Subject to regularity conditions, these unconditional moment conditions can then be used to estimate the model parameters $\theta_0$ by the generalized method of moments (GMM). This approach solves the incidental parameter problem \citep{neyman1948consistent}, because the moment condition \eqref{MomentUnconditional} does not feature the individual-specific effects $A$ at all. No assumptions are imposed on the distribution of those nuisance parameters, and they need not be estimated.

Finding a non-trivial solution $m_{y_0}(y,x,\theta)$ to the system \eqref{MomentConditionSum} is generally hard. The system must hold for all $\alpha \in \mathbb{R}^Q$, which imposes infinitely many restrictions on the $Q^T$ unknowns. For arbitrary parametric models, only the trivial solution $m_{y_0} \equiv 0$ may exist. Even when non-trivial solutions do exist, they are typically not available in closed form. \citet{bonhomme2012functional} provides an abstract characterization of the space of valid moment conditions through functional differencing, but this characterization does not by itself deliver explicit moment functions that can be used for estimation. As a result, explicit constructions of moment functions free of the fixed effects have so far been carried out on a model-by-model basis, relying on features specific to each setting, such as the logit structure in the binary and ordered choice models. The main contribution of this paper is to provide a similar explicit construction for the dynamic multinomial logit model, using the transformed model approach developed in the next subsection.

\subsection{Transformed model}
\label{Sec: Transformed model}

Our construction of valid moment functions is based on the concept of a transformed model. The key observation is that the fixed effects $A$ enter the model probabilities in a specific way, and by comparing the true model with a suitably constructed transformed model, one can eliminate the dependence on $A$.

The intuition behind this construction is as follows. 
The fixed effects $A$ enter the joint probability $p(y \mid y_0, x, \theta, \alpha)$ in two ways: (i) through the product $\prod_{t=1}^T \exp(\alpha_{y_t})$ in the numerators of the period-by-period choice probabilities, and (ii) through their denominators $\prod_{t=1}^T\sum_{r=1}^{Q}
\exp\left( z_t(r,y_{t-1},x,\theta) +\alpha_r \right)$. 
Crucially, the numerator product depends on the outcome sequence only through the unordered multiset $\{y_1, \ldots, y_T\}$, which renders it invariant under any permutation $\sigma$ of the time indices. This implies that comparing the true model with a transformed one that reshuffles outcomes while preserving the denominators will eliminate the fixed effect contribution. We will show below that the resulting likelihood ratio $p^\sigma(y \mid y_0,x,\theta,\alpha)/p(y \mid y_0,x,\theta,\alpha)$ is indeed 
free of $\alpha$ (Lemma~\ref{lemma:SnoAlpha}), and that when $\sigma$ is a cyclic permutation this ratio can be turned into valid moment conditions for estimation.

Let $\sigma: \{1,\ldots,T\} \to \{1,\ldots,T\}$ be a permutation of the time indices. We define the transformed model probability as
\begin{align}
p^\sigma(y \mid y_0, x, \theta, \alpha) :=  p_1(y_{\sigma(T)} \mid y_{0}, x, \theta, \alpha)  \, \prod_{t=2}^T \, p_t(y_{\sigma(t-1)} \mid y_{t-1}, x, \theta, \alpha) .
\label{AuxModel}
\end{align}
Compared to the true model, where period $t$ generates outcome $y_t$, the transformed model reassigns which outcome is generated in each period: period 1 generates $y_{\sigma(T)}$, and period $t \geq 2$ generates $y_{\sigma(t-1)}$. The true model is the special case $\sigma=\sigma_+$, where $\sigma_+$ is the cyclic increment
$\sigma_+(t)=t+1$ for $t<T$ and $\sigma_+(T)=1$.

The expression $p^\sigma$ in \eqref{AuxModel} is nonnegative by construction, but it does not sum to one over all $y \in \{1,\ldots,Q\}^T$ for arbitrary permutations $\sigma$. For a general permutation $\sigma$, the factors in \eqref{AuxModel} do not necessarily generate a single cyclic $T$-step Markov chain through all $T$ outcomes.
To see this formally, define
\begin{equation}
\tau(k) := \sigma^k(T), \qquad k=1,\ldots,T,
   \label{TauDefinition}
\end{equation}
where $\sigma^k$ is the $k$-fold composition of $\sigma$, e.g. $\sigma^2(T)= \sigma(\sigma(T))$.
The transformed model has a proper Markov structure only if $\tau$ visits each index in $\{1,\ldots,T\}$ exactly once. In this case, $\tau$ is itself a permutation of $\{1,\ldots,T\}$, and $p^\sigma$ can be rewritten as
\begin{align}
p^\sigma(y \mid y_0, x, \theta, \alpha) =
p_1(y_{\tau(1)} \mid y_0, x, \theta, \alpha) \prod_{k=2}^T p_{\tau(k-1)+1}(y_{\tau(k)} \mid y_{\tau(k-1)}, x, \theta, \alpha) .
\label{AuxModelTau}
\end{align}
We call the sequence $\tau = (\tau(1), \tau(2), \ldots, \tau(T))$ the generation order of $\sigma$. The name reflects the role that $\tau$ plays in the representation \eqref{AuxModelTau}: starting from $y_0$, the transformed model generates $y_{\tau(1)}$ first, then $y_{\tau(2)}$, and so on, with $y_{\tau(T)}$ generated last. In other words, the generation order of $\sigma$ is the order in which the outcomes $y_1, \ldots, y_T$ are produced by the transformed model $p^\sigma$ as a Markov chain.
It turns out that \eqref{AuxModel} can be rewritten as \eqref{AuxModelTau} if and only if $\sigma$ is a {\it cyclic} permutation of $\{1,\ldots,T\}$.
This is a well-known concept in combinatorics, which simply means that $T,\sigma(T),\sigma^2(T),\sigma^3(T),\ldots,\sigma^{T-1}(T)$ are all distinct.\footnote{
Formally, a permutation $\sigma:\{1,\ldots,T\}\to\{1,\ldots,T\}$ is cyclic if and only if the orbit of $T$ under $\sigma$ has length $T$. Equivalently,
$
\{\sigma^k(T): k=0,1,\ldots,T-1\}=\{1,\ldots,T\}.
$
}

To see why cyclic permutations are the relevant class, recall that every permutation $\sigma$ of $\{1,\ldots,T\}$ decomposes into disjoint cycles, and that the sequence $\tau(k) = \sigma^k(T)$ traces only the cycle of $\sigma$ containing $T$. The reformulation \eqref{AuxModelTau} represents the transformed model as a single Markov chain through all $T$ time periods, which is possible only when the cycle of $T$ has length $T$, that is, when $\sigma$ is cyclic. For a non-cyclic $\sigma$, the cycle of $T$ has length strictly less than $T$, and $\tau(k)$ returns to $T$ before visiting the time indices in the other cycles of $\sigma$. In this case $p^\sigma$ factors as a product of independent Markov chains, one per cycle of $\sigma$. The chains associated with cycles not containing $T$ are ``closed loops'' that start and end at the same outcome, and for generic $(x,\theta,\alpha)$ they do not sum to one over the corresponding outcomes. The cyclic case is therefore the only one for which $\sum_y p^\sigma(y \mid y_0, x, \theta, \alpha) = 1$ holds for all $(x,\theta,\alpha)$; see Appendix~\ref{app:generalized} for the factorization in the non-cyclic case.

Once we have the reformulation \eqref{AuxModelTau}, then it is straightforward to show that $p^\sigma(y \mid y_0, x, \theta, \alpha) $ sums to one, as summarized by the following lemma.

\begin{lemma}
\label{lemma:cyclic}
If $\sigma:\{1,\ldots,T\}\to\{1,\ldots,T\}$ is a cyclic permutation, then
\begin{align}
\sum_{y \in \{1,\ldots,Q\}^T} p^\sigma(y \mid y_0, x, \theta, \alpha) = 1 .
   \label{PSigmaSumToOne}
\end{align}
\end{lemma}

\begin{proof}[\bf Proof of Lemma~\ref{lemma:cyclic}]
Since $\sigma$ is cyclic,  the sequence $\tau(k)=\sigma^k(T)$ visits each index in $\{1,\ldots,T\}$ exactly once. Hence \eqref{AuxModel} is a $T$-step Markov chain written in a different time order, see \eqref{AuxModelTau}. Therefore iterated summation gives \eqref{PSigmaSumToOne}
since at each step $\sum_{j=1}^Q p_t(j\mid \cdot)=1$.
\end{proof}

\subsection{Moment construction} 
\label{Sec: Moment construction}
 
Let $\sigma:\{1,\ldots,T\}\to\{1,\ldots,T\}$ be a permutation. Define $\sigma(0):=\sigma(T)$. Then \eqref{AuxModel} can be written more concisely as
\begin{align}
p^\sigma(y \mid y_0, x, \theta, \alpha)
= \prod_{t=1}^T p_t\!\left(y_{\sigma(t-1)} \mid y_{t-1}, x, \theta, \alpha\right),
\label{AuxModelSimplified}
\end{align}
and this definition of $\sigma(0)$ will simplify many expressions below.

The ratio of the transformed model probability to the true model probability plays a central role in our construction. The following lemma shows that the ratio
$
 p^\sigma(y \mid y_0, x, \theta, \alpha) / {p(y \mid y_0, x, \theta, \alpha)}
$
is free of $\alpha$ and equal to
\begin{align}
S^\sigma_{y_0}(y, x, \theta)
:= \exp\!\left( \sum_{t=1}^T \Big[ z_t\!\left(y_{\sigma(t-1)}, y_{t-1}, x, \theta\right)
- z_t\!\left(y_t, y_{t-1}, x, \theta\right) \Big] \right).
\label{DefS}
\end{align}

\begin{lemma}
\label{lemma:SnoAlpha}
For any permutation $\sigma:\{1,\ldots,T\}\to\{1,\ldots,T\}$ we have
\begin{align}
\frac{p^\sigma(y \mid y_0, x, \theta, \alpha)}{p(y \mid y_0, x, \theta, \alpha)}
= S^\sigma_{y_0}(y, x, \theta),
\label{SExplicit}
\end{align}
implying that this ratio is free of $\alpha$.
\end{lemma}

\begin{proof}[\bf Proof of Lemma~\ref{lemma:SnoAlpha}]
Using \eqref{AuxModelSimplified} and the definition of the true model,
\[
\frac{p^\sigma(y \mid y_0,x,\theta,\alpha)}{p(y \mid y_0,x,\theta,\alpha)}
=\prod_{t=1}^T 
\frac{p_t\!\left(y_{\sigma(t-1)} \mid y_{t-1},x,\theta,\alpha\right)}
     {p_t\!\left(y_t \mid y_{t-1},x,\theta,\alpha\right)}.
\]
For each $t$, the normalizing denominators cancel, so
\[
\frac{p_t\!\left(y_{\sigma(t-1)} \mid y_{t-1},x,\theta,\alpha\right)}
     {p_t\!\left(y_t \mid y_{t-1},x,\theta,\alpha\right)}
=\exp\!\Big(z_t(y_{\sigma(t-1)},y_{t-1},x,\theta)-z_t(y_t,y_{t-1},x,\theta)\Big)
\,\exp\!\Big(\alpha_{y_{\sigma(t-1)}}-\alpha_{y_t}\Big).
\]
Multiplying over $t$ gives $S^\sigma_{y_0}(y, x, \theta)$ times
\[
\exp\!\left(\sum_{t=1}^T \alpha_{y_{\sigma(t-1)}}-\sum_{t=1}^T \alpha_{y_t}\right).
\]
Since $\sigma$ is a permutation and $\sigma(0)=\sigma(T)$, the indices $\sigma(t-1)$, $t=1,\ldots,T$, are a reordering of $1,\ldots,T$. Hence $\sum_{t=1}^T \alpha_{y_{\sigma(t-1)}}=\sum_{t=1}^T \alpha_{y_t}$, so the $\alpha$ term cancels. This yields \eqref{SExplicit}. 
\end{proof}

\bigskip
 
Next, let $\sigma$ be a cyclic permutation. Then the simplest valid moment condition that follows directly from Lemma~\ref{lemma:cyclic} and Lemma~\ref{lemma:SnoAlpha} is
\begin{align}
\mathbb{E}\left[ 1 - S^\sigma_{Y_0}(Y, X, \theta_0) \, \Big| \, Y_0, X, A \right] = 0 .
\label{SimpleMoment}
\end{align}
To see this, note that
\[
\mathbb{E}\!\left[S^\sigma_{Y_0}(Y,X,\theta_0)\mid Y_0,X,A\right]
=\sum_{y} S^\sigma_{Y_0}(y,X,\theta_0)\, p(y\mid Y_0,X,\theta_0,A)
=\sum_{y} p^\sigma(y\mid Y_0,X,\theta_0,A)
=1,
\]
where the second equality uses \eqref{SExplicit} and the last equality uses \eqref{PSigmaSumToOne}.

While valid, the moment function $1 - S^\sigma_{y_0}(y, x, \theta)$ may not be the most useful in practice, because it typically has $\theta = 0$ as a root regardless of the true parameter value.  Indeed, for the linear specification \eqref{LinearSpec} we have $z_t(\cdot,\cdot,x,0) = 0$, so $S^\sigma_{y_0}(y,x,0) = 1$ for every permutation and every outcome sequence, and the moment function $1 - S^\sigma$ vanishes identically at $\theta = 0$. Moment conditions based on \eqref{SimpleMoment} alone therefore cannot separate the true parameter value from zero. The following theorem therefore presents more general valid moment functions.

The construction works by pairing two cyclic permutations $\sigma_1$ and $\sigma_2$ whose generation orders $\tau_1$ and $\tau_2$ agree on their first $k-1$ entries, but differ at position $k$. For both of the transformed models, the first $k-1$ outcomes $Y_{\tau_1(1)}, \ldots, Y_{\tau_1(k-1)}$ are produced by identical Markov transitions, so conditioning on these outcomes removes the shared history. The divergence in position $k$ is what makes the moment function nontrivial. The outcomes in the remaining positions $k+1,\ldots,T$ telescope out of the expectation through the Markov factorization \eqref{AuxModelTau}, leaving an exact cancellation between the two $S^\sigma$ terms.

\begin{theorem}
\label{th:MainTheorem}
Let $\sigma_1$ and $\sigma_2$ be two cyclic permutations of $\{1,\ldots,T\}$, and define
$\tau_1(k) := \sigma_1^k(T)$ and $\tau_2(k) := \sigma_2^k(T)$ for $ k \in \{1,\ldots,T\}$, analogous to  \eqref{TauDefinition}.
Assume $\sigma_1 \neq \sigma_2$ and
let $k \in \{1,\ldots,T\}$ be the smallest integer such that $\tau_1(k) \neq \tau_2(k)$. Then for any $q \in \{1,\ldots,Q\}$, the moment function
\begin{align}
m^{\sigma_1,\sigma_2,q}_{y_0}(y, x, \theta) := \mathbbm{1}\{y_{\tau_1(k)} = q\} \, S^{\sigma_1}_{y_0}(y, x, \theta) - \mathbbm{1}\{y_{\tau_2(k)} = q\} \, S^{\sigma_2}_{y_0}(y, x, \theta)
\label{MainMoment}
\end{align}
satisfies
\begin{align}
\mathbb{E}\left[ m^{\sigma_1,\sigma_2,q}_{Y_0}(Y, X, \theta_0) \, \Big| \, Y_{\tau_1(1)}, Y_{\tau_1(2)}, \ldots, Y_{\tau_1(k-1)}, Y_0, X, A \right] = 0 .
\label{MainMomentCondition}
\end{align}
This implies that for any function $g:  \{1,\ldots,Q\}^{k-1} \rightarrow \mathbb{R}$ we have
\begin{align}
\mathbb{E}\left[ g( Y_{\tau_1(1)}, Y_{\tau_1(2)}, \ldots, Y_{\tau_1(k-1)}) \; m^{\sigma_1,\sigma_2,q}_{Y_0}(Y, X, \theta_0) \, \Big| \, Y_0, X, A \right] = 0 .
\label{MainMomentUnconditional}
\end{align}

\end{theorem}

\begin{proof}[\bf Proof of Theorem~\ref{th:MainTheorem}]
It suffices to prove \eqref{MainMomentCondition}, since \eqref{MainMomentUnconditional} follows by multiplying by $g(Y_{\tau_1(1)},\ldots,Y_{\tau_1(k-1)})$ and taking iterated expectations.

Fix arbitrary $q_1,\ldots,q_{k-1} \in \{1,\ldots,Q\}$. For $i \in \{1,2\}$, Lemma~\ref{lemma:SnoAlpha} gives $S^{\sigma_i} = p^{\sigma_i}/p$, so
\begin{align*}
&\mathbb{E}\!\left[\mathbbm{1}\{Y_{\tau_i(k)}=q\}\,S^{\sigma_i}_{Y_0}(Y,X,\theta_0)\cdot \prod_{j=1}^{k-1}\mathbbm{1}\{Y_{\tau_1(j)}=q_j\} \,\Big|\, Y_0,X,A\right] \\
&\qquad \qquad \qquad \qquad \qquad = \sum_{\substack{y \in \{1,\ldots,Q\}^T:\\ y_{\tau_1(j)}=q_j\ \forall\, j=1,\ldots,k-1}} \mathbbm{1}\{y_{\tau_i(k)}=q\}\; p^{\sigma_i}(y\mid Y_0,X,\theta_0,A) .
\end{align*}
By the Markov factorization \eqref{AuxModelTau}, summing over the remaining components $y_{\tau_i(k+1)},\ldots,y_{\tau_i(T)}$ telescopes to one, and the right-hand side equals
\[
p_1(q_1 \mid Y_0,X,\theta_0,A)\,
\prod_{j=2}^{k-1} p_{\tau_i(j-1)+1}(q_j \mid q_{j-1},X,\theta_0,A)\,
\cdot\, p_{\tau_i(k-1)+1}(q \mid q_{k-1},X,\theta_0,A).
\]
For $k=1$ no outcomes are held fixed, and the right-hand side is simply $p_1(q \mid Y_0,X,\theta_0,A)$.
Since $\tau_1(j)=\tau_2(j)$ for all $j \le k-1$, this expression is the same for $i=1$ and $i=2$. Hence the two terms in \eqref{MainMoment} cancel. Since this holds for all $(q_1,\ldots,q_{k-1})$, equation \eqref{MainMomentCondition} follows.
\end{proof}

\subsubsection{Example: moment conditions for $T=3$}

To illustrate, consider $T=3$. There are $(T-1)!=2$ cyclic permutations of $\{1,2,3\}$: the cyclic increment $\sigma_+$, which gives the true model, and the reverse cycle $\sigma_1$ defined by $\sigma_1(1)=3$, $\sigma_1(2)=1$, $\sigma_1(3)=2$. Taking $\sigma_2 = \sigma_+$, the generation orders are $\tau_1=(2,1,3)$ and $\tau_2=(1,2,3)$, which first differ at $k=1$. Since $\sigma_2=\sigma_+$ is the true model, $S^{\sigma_2}_{y_0}(y,x,\theta)=1$, and Theorem~\ref{th:MainTheorem} yields $Q$ conditional moment conditions:
\begin{align}
\mathbb{E}\Big[\mathbbm{1}\{Y_2=q\}\,S^{\sigma_1}_{Y_0}(Y,X,\theta_0) - \mathbbm{1}\{Y_1=q\} \,\Big|\, Y_0, X, A\Big] = 0, \quad q=1,\ldots,Q,
\label{MomentT3}
\end{align}
where, by \eqref{DefS},
\begin{align}
S^{\sigma_1}_{y_0}(y,x,\theta) = \exp\!\Big(&z_1(y_2,y_0,x,\theta) - z_1(y_1,y_0,x,\theta) \notag\\
&+ z_2(y_3,y_1,x,\theta) - z_2(y_2,y_1,x,\theta) \notag\\
&+ z_3(y_1,y_2,x,\theta) - z_3(y_3,y_2,x,\theta)\Big).
\label{ST3}
\end{align}
For $Q=2$, the moment conditions \eqref{MomentT3} recover the leading result for the dynamic binary logit model with $T=3$ in \cite{honore2025moment}. The transformed model approach provides a particularly transparent derivation of these known moment conditions for $Q=2$. %

More generally, varying $(\sigma_1, \sigma_2, q, g)$ over all admissible choices in Theorem~\ref{th:MainTheorem} generates a large family of valid moment functions, many of which are linearly dependent. Determining the number of linearly independent moment conditions that this construction yields, as a function of $Q$ and $T$, is the subject of the next subsection.

\subsection{Counting of linearly independent moment conditions}
\label{Sec: Counting}

Theorem~\ref{th:MainTheorem} produces a large family of moment functions as we vary the cyclic permutations $\sigma_1, \sigma_2$ and the weighting function $g$. Many choices will generate moment conditions that are redundant, in the sense that they can be written as linear combinations of others in the family. It is therefore useful to know how many linearly independent moment conditions can be obtained from the theorem.

We now describe an explicit parameterization that extracts $\ell(Q, T)$ moment functions from Theorem~\ref{th:MainTheorem}, where $\ell(Q, T)$ is the count given in \eqref{CountingFormula} below. The moment functions will be indexed by tuples $(j, \mathbf{t}, s_2, \mathbf{q})$ where $j \in \{0, 1, \ldots, T - 3\}$ is an integer that records the number of shared time indices between the two generation orders of Theorem~\ref{th:MainTheorem}, $\mathbf{t} = (t_1 < \cdots < t_j)$ is a collection of $j$ time periods in $\{1, \ldots, T - 1\}$, $s_1 := \min(\{1, \ldots, T - 1\} \setminus \mathbf{t})$ is therefore determined by $\mathbf{t}$, and $s_2$ is any element of $\{1, \ldots, T - 1\} \setminus (\mathbf{t} \cup \{s_1\})$, and $\mathbf{q} = (q_0, q_1, \ldots, q_j) \in \{1, \ldots, Q\}^{j+1}$ is a vector of outcome values. For each such tuple the corresponding moment function takes the form
\begin{align}
 m^{j, \mathbf{t}, s_1, s_2, \mathbf{q}}_{y_0}(y, x, \theta) \;:=\; \left[ \prod_{l=1}^j  \mathbbm{1}\!\left\{y_{t_l} = q_l \right\} \right] \Big[\mathbbm{1}\{y_{s_1} = q_0\} \, S^{\sigma_1}_{y_0}(y, x, \theta) - \mathbbm{1}\{y_{s_2} = q_0\} \, S^{\sigma_2}_{y_0}(y, x, \theta)\Big],
\label{MomentParameterized}
\end{align}
where $\sigma_1$ and $\sigma_2$ are the two cyclic permutations of $\{1, \ldots, T\}$ whose generation orders agree on the first $j$ entries $(t_1, \ldots, t_j)$ but differ at position $j + 1$, with $\sigma_1$ placing $s_1$ at that position and $\sigma_2$ placing $s_2$.\footnote{Explicitly, $\sigma_1$ has generation order $\tau_1 = (t_1, \ldots, t_j, s_1, r_1, \ldots, r_{T - j - 3}, s_2, T)$ and $\sigma_2$ has generation order $\tau_2 = (t_1, \ldots, t_j, s_2, r_1, \ldots, r_{T - j - 3}, s_1, T)$, where $r_1 < \cdots < r_{T - j - 3}$ are the remaining elements of $\{1, \ldots, T - 1\} \setminus (\mathbf{t} \cup \{s_1, s_2\})$ in increasing order.} Since the generation orders first differ at position $k = j + 1$, Theorem~\ref{th:MainTheorem} applies with this $k$, $q = q_0$, and weighting function $g(y_{t_1}, \ldots, y_{t_j}) = \prod_{l=1}^j \mathbbm{1}\{y_{t_l} = q_l\}$, and guarantees that $m^{j, \mathbf{t}, s_1, s_2, \mathbf{q}}_{y_0}$ satisfies \eqref{MomentCondition} at $\theta = \theta_0$.

To avoid counting redundant moment functions, we further impose the restriction $t_r = r$ for $r = 1, \ldots, \max(0, j - (Q - 2))$ on $\mathbf{t}$.\footnote{Without this restriction the construction would yield more tuples than $\ell(Q, T)$, but those additional moment functions would be linearly dependent on those in the constrained set. For $j \leq Q - 2$ the restriction is vacuous; for $j > Q - 2$ it fixes the first $j - (Q - 2)$ entries of $\mathbf{t}$ to the smallest possible values $1, 2, \ldots, j - (Q - 2)$.}

Counting the distinct tuples $(j, \mathbf{t}, s_2, \mathbf{q})$ that satisfy all the restrictions above then gives\footnote{For each $j$, the restriction on $\mathbf{t}$ fixes $\max(0, j - (Q - 2))$ of its entries and leaves $\min(j, Q - 2)$ entries to be chosen freely from the $T - 1 - \max(0, j - (Q - 2))$ remaining positions in $\{1, \ldots, T - 1\}$, giving $\binom{T - 1 - \max(j - (Q - 2),\, 0)}{\min(Q - 2,\, j)}$ choices for $\mathbf{t}$. Given $\mathbf{t}$, there are $T - 2 - j$ choices for $s_2 \in \{1, \ldots, T - 1\} \setminus (\mathbf{t} \cup \{s_1\})$, and $Q^{j+1}$ choices for $\mathbf{q}$. Summing the product of these three counts over $j \in \{0, 1, \ldots, T - 3\}$ yields \eqref{CountingFormula}.}
\begin{align}
\ell(Q, T) \;=\; \sum_{j=0}^{T - 3} \binom{T - 1 - \max(j - (Q-2),\, 0)}{\min(Q - 2,\, j)} \cdot (T - 2 - j) \cdot Q^{j + 1}.
\label{CountingFormula}
\end{align}

For small values of $Q$ this simplifies to $\ell(2, T) = 2^T - 2T$, $\ell(3, T) = 3^T - \binom{3T}{2} + 6(T - 1)$, and $\ell(4, T) = 4^T - \binom{4T}{3} + 40(T - 1)^2$. More generally, for $Q \geq T - 1$, it becomes $\ell(Q, T) = Q^T - Q (2 + Q - T) (1 + Q)^{T - 2}$. In each case the formula has the structure $\ell(Q, T) = Q^T - r(Q, T)$, where the ambient dimension $Q^T$ counts all real-valued functions of the outcome sequence $y$, and the subtracted term $r(Q, T)$ counts the linear restrictions on $m$ imposed by requiring the moment condition \eqref{MomentConditionSum} to hold for all values of the fixed effect $\alpha$.

Numerical investigations and the partial analytical results discussed in Appendix~\ref{app:counting} lead us to the following conjecture.

\begin{conjecture}
\label{conj:counting}
For $T \geq 3$ and $Q \geq 2$ and for generic $(y_0, x, \theta_0)$, the $\ell(Q, T)$ moment functions in \eqref{MomentParameterized} are linearly independent, and every valid moment function of the form \eqref{MomentCondition} at $\theta = \theta_0$ is a linear combination of them.
\end{conjecture}

The status of the conjecture is as follows. For $Q = 2$ the lower bound follows from the explicit moment functions of \citet{honore2025moment}, and the matching upper bound is established in subsequent work; see \citet{dobronyi2021identification} and \cite{dano2023transition}. For $T = 3$ and any $Q \geq 2$, we prove in Appendix~\ref{app:counting} that the $\ell(Q, 3) = Q$ moment functions in \eqref{MomentParameterized} are linearly independent, so the dimension of the space of valid moment conditions is at least $\ell(Q, 3)$. For all combinations of $(Q, T)$ reported in Table~\ref{tab:counting}, the conjecture has been verified by direct numerical computation. A general analytical proof of the conjecture remains an open problem. Appendix~\ref{app:counting} discusses two natural proof strategies (a polynomial expansion of the likelihood and a direct linear-independence argument on the moments in \eqref{MomentParameterized}) and identifies the technical obstacle.

The qualifier ``generic'' in the conjecture means that it is expected to hold outside a measure-zero subset of $(Y_0, X, \theta_0)$. For special values of $X$ and $\theta_0$, additional moment conditions may become available, so the total number of linearly independent moment conditions can exceed $\ell(Q, T)$.

\begin{table}[t]
\begin{center}
\begin{tabular}{c|cccccccc}
     & $T=1$ & $T=2$ & $T=3$ & $T=4$ & $T=5$ & $T=6$ & $T=7$ & $T=8$ \\
\hline
$Q=2$ & 0 & 0 & 2 & 8 & 22 & 52 & 114 & 240 \\
$Q=3$ & 0 & 0 & 3 & 33 & 162 & 606 & 2013 & 6327 \\
$Q=4$ & 0 & 0 & 4 & 56 & 524 & 3072 & 14548 & 62536 \\
$Q=5$ & 0 & 0 & 5 & 85 & 965 & 9145 & 62500 & 358280 \\
$Q=6$ & 0 & 0 & 6 & 120 & 1602 & 17844 & 179094 & 1399680 \\
$Q=7$ & 0 & 0 & 7 & 161 & 2471 & 31633 & 364791 & 3929793 \\
$Q=8$ & 0 & 0 & 8 & 208 & 3608 & 52192 & 679976 & 8274160 \\
\end{tabular}
\end{center}
\caption{\label{tab:counting} Conjectured number of linearly independent moment conditions $\ell(Q, T)$ for various values of $Q$ (the number of alternatives) and $T$ (the number of time periods after the initial condition). Verified by numerical computation for all entries in the table.}
\end{table}

Table~\ref{tab:counting} reports $\ell(Q, T)$ for combinations relevant for empirical applications. Two features are worth noting. First, $\ell(Q, T) = 0$ for $T \leq 2$ for every $Q$, which confirms that at least three time periods after the initial condition are required for moment conditions free of the fixed effects to exist for generic parameters. Second, the count grows quickly in both $Q$ and $T$: with $Q = 4$ alternatives and $T = 5$ time periods, the parameterization yields 524 linearly independent moment conditions for each value of the initial condition. Only a small subset of these moment conditions is typically used for estimation in practice. We return to this issue in Section~\ref{sec:empirical}.

\section{Moment conditions for more general models}
\label{sec:moregen}

Most of the derivations in Sections \ref{Sec: Transformed model} and \ref{Sec: Moment construction} are generic and do not
use the specific structure of the multinomial logit model. The only
exception is Lemma \ref{lemma:SnoAlpha} which establishes that for any permutation $\sigma
:\{1,\ldots ,T\}\rightarrow \{1,\ldots ,T\}$, 
\begin{equation}
S_{y_{0}}^{\sigma }(y,x,\theta ):=\frac{p^{\sigma }(y\mid y_{0},x,\theta
,\alpha )}{p(y\mid y_{0},x,\theta ,\alpha )}  \label{SExplicitGeneric}
\end{equation}%
is free of $\alpha $. In this section, we characterize a larger class of
models for which this is true, and it follows that moment conditions can be
constructed using the arguments of Sections \ref{Sec: Transformed model} and \ref{Sec: Moment construction} for any
model in this larger class.

Consider the model for the joint distribution of $\left( Y_{1},Y_{2},\ldots
,Y_{T}\right) $:

\begin{equation}
p(y\mid y_{0},x,\theta ,\alpha )=P\left( Y=y\,\Big|\,Y_{0}=y_{0},x,A=\alpha
\right) =\prod_{t=1}^{T}f(y_{t},z_{t}\left( y_{t-1},x,\theta \right) ,\alpha
)  \label{JointProbGeneric}
\end{equation}%
for functions $f$ and $z_{1},\ldots ,z_{T}$.

\begin{lemma}
\label{lemma:SnoAlphaGeneric}
Let $p$ be defined as in (\ref{JointProbGeneric}). Assume that for every
$(\eta _{1},\ldots ,\eta _{T})$ in the range of the functions $z_{t}$, the
joint distribution
\[
\prod_{t=1}^{T}f(y_{t},\eta _{t},\alpha )
\]%
on the outcome space of $(Y_{1},\ldots ,Y_{T})$ admits a statistic $S=s(Y_{1},\ldots ,Y_{T})$ such
that the conditional distribution of $(Y_{1},\ldots ,Y_{T})$ given $S$ does
not depend on $\alpha $, and that $s$ is symmetric in its arguments (so that
$s(y_{\pi (1)},\ldots ,y_{\pi (T)})=s(y_{1},\ldots ,y_{T})$ for every
permutation $\pi $ of $\{1,\ldots ,T\}$). Then for any permutation $\sigma
:\{1,\ldots ,T\}\rightarrow \{1,\ldots ,T\}$, the ratio
\[
\frac{p^{\sigma }(y\mid y_{0},x,\theta ,\alpha )}{p(y\mid y_{0},x,\theta
,\alpha )}=S_{y_{0}}^{\sigma }(y,x,\theta )
\]%
is free of $\alpha $.
\end{lemma}

\begin{proof}[\textbf{Proof of Lemma~\protect\ref{lemma:SnoAlphaGeneric}}]
The assumptions imply that for any values of $\eta _{1},\ldots ,\eta _{T}$,
the ratio $\frac{\prod_{t=1}^{T}f(s_{t},\eta _{t},\alpha )}{%
\prod_{t=1}^{T}f(y_{t},\eta _{t},\alpha )}$ does not depend on $\alpha $
provided that $\left( s_{1},\ldots ,s_{T}\right) $ is a
permutation of $\left( y_{1},\ldots ,y_{T}\right) $. With $\eta
_{t}=z_{t}\left( y_{t-1},x,\theta \right) $, we then have 
\[
\frac{p^{\sigma }(y\mid y_{0},x,\theta ,\alpha )}{p(y\mid y_{0},x,\theta
,\alpha )}=\frac{f(y_{\sigma (T)}, \eta _{1},\alpha
)\,\prod_{t=2}^{T}\,f(y_{\sigma (t-1)}, \eta _{t},\alpha )}{%
\prod_{t=1}^{T}f(y_{t},\eta _{t},\alpha )}.
\]
Since $\left(y_{\sigma (T)},y_{\sigma (1)},\ldots ,y_{\sigma (T-1)} \right)$ is a permutation of $\left( y_{1},\ldots ,y_{T}\right) $, this does not depend on $\alpha$.
\end{proof}

Combining Lemma \ref{lemma:SnoAlphaGeneric}  with the results in Sections \ref{Sec: Transformed model} and \ref{Sec: Moment construction} implies that for any static model for which there exists a symmetric sufficient statistic for $A$, it is possible to construct moment conditions for the corresponding dynamic model in which $Y_{t-1}$ is an explanatory variable in period $t$. This unifies some known results and makes  it clear how to construct moment conditions in models that have not previously been considered in the literature. Moreover, the logic behind the construction also applies to continuous distributions. We provide some examples below.

\subsection{Example: exponential--family panel models (\citealt{hahn1997note})}

A rich class of models satisfying the hypotheses of Lemma~\ref{lemma:SnoAlphaGeneric} is provided by the exponential--family panel models studied by \citet{hahn1997note}. In the notation of \eqref{JointProbGeneric}, the period--$t$ density takes the form
\begin{equation}
f(y_t,\eta,\alpha)=\exp\!\big(\,s(y_t)'\,\alpha\,e(\theta)+d(\eta,\alpha,\theta)+b(y_t,\eta,\theta)\,\big),
\label{eq:hahn-density}
\end{equation}
where $\alpha$ is the fixed effect, $\eta=z_t(y_{t-1},x,\theta)$ is the index, $s(\cdot)$ is a vector of canonical statistics, $e(\theta)$ is a known function of the common parameter, $d(\cdot)$ is the log--normalizing constant (which does not involve $y_t$), and $b(\cdot)$ collects the remaining terms, which involve $y_t$ and the index but not the fixed effect. This is the model in equation~(1) of \citet{hahn1997note}, with the role of the regressors played here by the index $\eta$, and as \citet{hahn1997note} notes, this class contains the panel linear regression model with normal errors, the panel logit model, the panel Poisson regression model, and the panel negative binomial model.

Consider the static case where $\eta=z_t(y_{t-1},x,\theta)$ does not depend on $y_{t-1}$.  In this case, 
\[
S(Y)=\sum_{t=1}^{T}s(Y_t)
\]
is a sufficient statistic for $\alpha$, and it is symmetric in $(Y_1,\ldots,Y_T)$ (Lemma~\ref{lem:general-suff} of Appendix~\ref{app:moregen-details}). The assumptions of Lemma~\ref{lemma:SnoAlphaGeneric} are therefore satisfied, and the moment--condition construction of Sections~\ref{Sec: Transformed model}--\ref{Sec: Moment construction} applies to a model that includes lagged dependent variables as explanatory variables. 

To obtain the moment functions explicitly, form the ratio $p^{\sigma}/p$ using \eqref{eq:hahn-density}. The normalizing constants $d(\eta_t,\alpha,\theta)$ cancel period by period; and because $\sum_{t}s(y_{\sigma(t-1)})=\sum_{t}s(y_t)$, the fixed effect term $s(y_t)'\alpha\,e(\theta)$ cancels as well. What remains is
\begin{equation}
S^{\sigma}_{y_0}(y,x,\theta)=\exp\!\Big(\sum_{t=1}^{T}\big[\,b\big(y_{\sigma(t-1)},\eta_t,\theta\big)-b\big(y_t,\eta_t,\theta\big)\,\big]\Big),
\qquad \eta_t=z_t(y_{t-1},x,\theta),
\label{eq:hahn-Sratio}
\end{equation}
which is free of $\alpha$, in agreement with Lemma~\ref{lemma:SnoAlphaGeneric}.  Exactly as in Theorem~\ref{th:MainTheorem}, take $T=3$, $\sigma_2=\sigma_+$ (so that $S^{\sigma_2}\equiv1$), and $\sigma_1$ the reverse cycle; when $Y$ is continuous the indicator $\mathbbm{1}\{Y_{\tau(k)}=q\}$ is replaced by an arbitrary bounded function $g$, and the moment condition is that, for every such $g$,
\begin{equation}
\mathbb{E}\!\left[\,g(Y_2)\,S^{\sigma_1}_{Y_0}(Y,X,\theta_0)-g(Y_1)\;\Big|\;Y_0,X,A\,\right]=0 .
\label{eq:hahn-moment}
\end{equation}
The validity of \eqref{eq:hahn-moment} is established more formally in Proposition~\ref{prop:general-moment} of Appendix~\ref{app:moregen-details}.

\subsection{Example: sufficient statistics that depend on $\theta$}

The model in \eqref{eq:hahn-density} implies a sufficient statistic that is free of the common parameter and the regressors. However, the construction of moment conditions suggested above requires only that the static model admits a sufficient statistic for $A$ that is symmetric in the time indices. It therefore applies more broadly, including cases in which the relevant statistic depends on $\theta$, so long as symmetry across $t$ is preserved. A bivariate gamma panel model in the spirit of \citet{lindsay1982conditional} illustrates this. For each $t$ the outcome is the pair $Y_t=(Y_{1t},Y_{2t})$ where, conditional on the past, the regressors, and the fixed effect $A$, the two components are independent with\footnote{These  use the common shape--rate parameterization.}
\begin{equation}
Y_{1t}\sim\mathrm{Gamma}\!\left(z_{1t},\,A\right),\qquad Y_{2t}\sim\mathrm{Gamma}\!\left(z_{2t},\,\phi A\right),
\label{eq:lindsay-model}
\end{equation}
 where $z_{jt}=z_{jt}(y_{t-1},X,\theta)$ for $j=1,2$ and $\theta=(\beta_1',\beta_2',\phi)'$. In the static case where the $z_{jt}(y_{t-1},X,\theta)$ do not depend on $y_{t-1}$,
\[
s_\theta(Y)=\sum_{t=1}^{T} Y_{1t}+\phi\sum_{t=1}^{T} Y_{2t}
\]
is a sufficient statistic for $A$ in the sense that the distribution of $(Y_1,\ldots,Y_T)$ conditional on $s_\theta(Y)$ does not depend on $A$. It is symmetric in the time indices, but in contrast to \eqref{eq:hahn-density}, it depends on the common parameter $\theta$ through $\phi$, so the model lies outside the class of \citet{hahn1997note}. Symmetry is nonetheless all that our moment conditions construction requires, and the same cancellation gives
\begin{equation}
S^{\sigma}_{y_0}(y,X,\theta)=\exp\!\Big(\sum_{t=1}^{T}\big[\,z_{1t}\,(\log y_{1,\sigma(t-1)}-\log y_{1t})+z_{2t}\,(\log y_{2,\sigma(t-1)}-\log y_{2t})\,\big]\Big),
\label{eq:lindsay-Sratio}
\end{equation}
free of $A$, so that, for any bounded $g$ on the outcome space,
\begin{equation}
\mathbb{E}\!\left[\,g(Y_2)\,S^{\sigma_1}_{Y_0}(Y,X,\theta_0)-g(Y_1)\;\Big|\;Y_0,X,A\,\right]=0;
\label{eq:lindsay-moment}
\end{equation}
see Appendix~\ref{app:moregen-details} for details. The cross--equation parameter $\phi$ enters only through the fixed effect--coupled rate of $Y_{2t}$ and hence cancels from $S^{\sigma}$ together with $A$. The construction thus delivers moment conditions for the index parameters $(\beta_1,\beta_2)$, while identifying $\phi$ would require other arguments.\footnote{For example, $\phi$ may be recovered as $\phi=\mathbb{E}\left[\frac{Y_{1t}}{z_{1t}}\right]/\mathbb{E}\left[\frac{Y_{2t}}{z_{2t}}\right]$ if  $\beta_{1},\beta_{2}$ are identified. This readily follows from the expression of the mean for a Gamma distribution.} Of course  in this case, different moment conditions for $(\beta_1,\beta_2)$ could have also been constructed using the marginal distributions for $Y_{1t}$ and $Y_{2t}$ and the fact that the fixed effect enters the means of $Y_{1t}$ and $Y_{2t}$ multiplicatively. 

\subsection{Example: linear panel data models with normal errors}
\label{sec: Linear Panel Data Models with Normal Errors}

The linear panel data model with normal errors and additive fixed effects,
\begin{equation}
Y_t = z_t(Y_{t-1},X,\theta) + A + \varepsilon_t,\qquad \varepsilon_t\mid (Y^{t-1},X,A)\sim\mathcal{N}(0,\sigma^2),
\label{eq:linear-model}
\end{equation}
is also a special case of \eqref{eq:hahn-density}, with the fixed effect entering through the canonical statistic $s(y_t)=y_t$. Hence $S(Y)=\sum_{t=1}^{T} Y_t$ is a symmetric sufficient statistic for $A$, and Lemma~\ref{lemma:SnoAlphaGeneric} applies. In forming $p^{\sigma}/p$, the terms in $y_t^2$, which do not involve the index, cancel by symmetry along with the fixed effect, leaving
\begin{equation}
S^{\sigma}_{y_0}(y,x,\theta)=\exp\!\Big(\tfrac{1}{\sigma^2}\sum_{t=1}^{T} z_t(y_{t-1},x,\theta)\,\big(y_{\sigma(t-1)}-y_t\big)\Big),
\label{eq:linear-Sratio}
\end{equation}
and the moment condition $\mathbb{E}\big[g(Y_2)\,S^{\sigma_1}_{Y_0}(Y,X,\theta_0)-g(Y_1)\mid Y_0,X,A\big]=0$ for any bounded $g$.

This example is instructive because it exposes the limits of the moment condition construction above. The model \eqref{eq:linear-model} already admits an obvious and well--behaved set of fixed effect-free moment conditions: because $Y_t-z_t=A+\varepsilon_t$, any contrast that differences out the fixed effect, such as
\begin{equation}
\mathbb{E}\big[(Y_t-z_t)-(Y_{t-1}-z_{t-1})\mid Y_0,X,A\big]=0,
\label{eq:linear-differencing}
\end{equation}
is linear in the data, possesses moments of all orders, and is the basis for the standard linear panel data estimators \citep{ArellanoBond1991}. When $z_t$ is linear in the lagged outcome, the moment function \eqref{eq:linear-Sratio} produced by our construction is, by contrast, the exponential of a quadratic \footnote{Interestingly, using a different approach \cite{Kitazawa2001} derived moment conditions for a model like (\ref{eq:linear-model}) that are based on exponential of linear functions of Gaussian outcomes.} form in the Gaussian outcomes --- the quadratic arising from the lagged outcome inside $z_t$ --- and is consequently heavy--tailed. Its conditional variance need not exist: it is finite only when the state--dependence coefficient lies in a bounded interval, and successively smaller intervals are required for the higher moments,\footnote{\label{Footnote: Existence of moments}For example, with $z_t(y_{t-1},x,\theta)=\gamma\,y_{t-1}+x_t'\beta$ and $T=3$, the conditional second moment $\mathbb{E}\big[(S^{\sigma_1}_{Y_0})^2\mid Y_0,X,A\big]$ is finite if and only if $\gamma$ lies between the two real roots of $1-4\gamma^2+4\gamma^3-4\gamma^4$, that is, $\gamma\in(-0.400,\,0.575)$; more generally the $k$--th conditional moment is finite if and only if $2k(k-1)\,\gamma^2(1-\gamma+\gamma^2)<1$.} and outside these ranges the moment has infinite variance and the associated estimator is no longer $\sqrt{n}$--asymptotically normal. Even where the variance is finite, the moment is strongly right--skewed, so generalized method of moments estimators based on it can be very imprecise.
We conclude from this example that the transformed--model approach can sometimes best be understood as a device for generating valid fixed effect-free moment conditions when no problem--specific reduction is apparent. The moment functions it delivers are always valid, but, as \eqref{eq:linear-Sratio} illustrates, they need not be well behaved, and they may be dominated by simple alternatives such as the differencing moments \eqref{eq:linear-differencing} that a model--specific analysis would suggest. The construction therefore complements, but does not substitute for, case--specific reasoning about how best to eliminate the fixed effects.

\section{Identification}
\label{sec:identification}

In earlier work, \citet{honore2000panel} studied the dynamic multinomial logit model with the linear specification \eqref{LinearSpec} and showed that 
$\theta_0$ can be identified with only $T=3$ periods. The key restriction they require is that 
the regressors do not vary over the last two periods, $X_2=X_3$. This makes it possible 
to construct a conditional likelihood that eliminates the fixed effects and leads to an 
elegant identification argument.\footnote{It can be shown that under the restriction $X_2=X_3$, model \eqref{JointProb} admits additional valid moment functions beyond those presented in Theorem \ref{th:MainTheorem}. The score of the conditional likelihood in \citet{honore2000panel} involves a particular linear combination of these extra moments. Here, our argument relies only on the moment conditions in Theorem \ref{th:MainTheorem}, which remain valid without imposing $X_2=X_3$ and can therefore be used for $\sqrt{n}$-consistent estimation with continuous covariates.}

The drawback is that this restriction is demanding when covariates are continuous, as in our 
empirical application which features product-specific prices. In that case, the event $X_2=X_3$ has probability zero, so estimation  relies only on the few observations for which $X_2$ and $X_3$ are close, implying in turn a slower-than-$\sqrt{n}$ rate of convergence.

This section shows that this shortcoming can be avoided by leveraging the fixed effect free moment condition in Theorem \ref{th:MainTheorem}.  We illustrate this point in the context of our empirical specification
\eqref{EmpiricalSpec},
\begin{align*}
z_t(q,Y_{i,t-1},X_i,\theta)
=
\gamma_q \mathbbm{1}\{Y_{i,t-1}=q\}
+
X_{it,q}\beta ,
\end{align*}
with $T=3$ periods. The target parameter is $\theta_{0}=(\beta_{0},\gamma_{01},\ldots,\gamma_{0Q})\in \Theta \subset \mathbb{R}^{Q+1}$ and Proposition \ref{prop:Identification_Habit} shows that it is uniquely identified under support conditions that do not require constancy of the covariates across certain periods. We impose the following assumptions. 
\begin{assumption}{\textsc{(Rectangular Compact Parameter Space)}.} \label{ass:Compact_Theta}
The parameter space is
\begin{align*}
    \Theta = [\underline{\beta},\overline{\beta}]\times [\underline{\gamma},\overline{\gamma}]^{Q}
\end{align*}
for scalars $\underline{\beta}<0<\overline{\beta}$, $\underline{\gamma}<0<\overline{\gamma}$.
\end{assumption}

\begin{assumption}{\textsc{(Identification Restrictions)}.} \label{ass:Identification_Restrictions}
The outcomes $Y=(Y_1,Y_2,Y_3)$ are generated according to model \eqref{JointProb} with $z_{.}(.,.,.,.)$ defined in \eqref{EmpiricalSpec} for $Q\geq 3$,\footnote{We focus our attention on the case $Q\geq 3$ as \cite{honore2025moment} already provided sufficient conditions for identification in the binary case $Q=2$.} $T=3$, and true parameter $\theta_{0}=(\beta_{0},\gamma_{01},\ldots,\gamma_{0Q})\in \Theta$. Moreover, for any initial condition $y_{0}\in \{1,\ldots,Q\}$, there exists a non-empty set $\mathcal{X}_{y_0}^{reg}\subseteq \mathbb{R}^{3Q}$ and a compact set $\mathcal{A}^{reg}_{y_{0}}$ such that for all $x\in \mathcal{X}_{y_0}^{reg}$, \\
$P(A\in \mathcal{A}^{reg}_{y_{0}} |Y_{0}=y_{0},X=x)$ is well defined and bounded away from zero.
\end{assumption}

\noindent Assumption \ref{ass:Compact_Theta} is standard. Assumption 
\ref{ass:Identification_Restrictions} imposes the regularity conditions employed in our 
identification argument. In particular, it guarantees that, the model probabilities of all choice 
histories $y=(y_1,y_2,y_3)$ remain strictly positive for each initial condition $y_0$.
 \\
\indent Our identification strategy exploits the conditional moment functions in \eqref{MainMoment}. With  $T=3$ periods, they take the form
 \begin{align} \label{moment_exp}
    m_{y_0}^{\sigma_{1},\sigma_{+},q}(y,x,\theta) \equiv m_{y_0}^{\sigma_{1},q}(y,x,\theta) =
  \mathbbm{1}\big\{ y_2 = q \big\}
S^{\sigma_{1}}_{y_0}(y,x,\theta)- \mathbbm{1}\big\{ y_1 = q \big\}, 
\end{align}
where $\sigma_{+}$ is the cyclical increment and $\sigma_{1}$ is the reverse cycle.  The core idea behind Proposition \ref{prop:Identification_Habit} is to craft a just-identified system of moment equations $\Xi(\theta)$ based on  \eqref{moment_exp} whose Jacobian $\nabla\Xi(\theta)$ forms a $P$-matrix on the whole parameter space $\Theta$; that is, a matrix with all principal minors positive. This property allows us to invoke a global univalence theorem of \citet{gale1965jacobian}, which implies that $\Xi(\theta)$ is injective on $\Theta$. Since the moment equations are satisfied at the true parameter, injectivity then yields global identification of $\theta_0$ as the unique root of $\Xi(\theta)=0$. \\
\indent The implementation of this idea requires introducing the following covariate regions. Fix a constant 
$B\in\mathbb{R}_{++}$ with $B>1/2$. For each $q\in\{1,\ldots,Q\}$ and $\epsilon>0$, let
\begin{align*}
    \mathcal{X}^{\gamma}_{q,\epsilon}
    =
    \Big\{
    x\in \mathcal{X}_{q}^{reg}:
    &\ \norm{x}_{\infty}< B, \\
    &\ x_{1q}\geq x_{3q}+\frac{\epsilon}{3}, 
       \quad x_{3q}\geq x_{2q}+\frac{\epsilon}{3}, \\
    &\ x_{2r}\geq x_{3r}+\frac{\epsilon}{3}, 
       \quad x_{3r}\geq x_{1r}+\frac{\epsilon}{3},
       \quad r\in\{1,\ldots,Q\}\setminus\{q\}, \\
    &\ \abs{x_{ts,l}}<\epsilon,
       \quad \forall (t,s,l)\in\{1,2,3\}^{2}\times\{1,\ldots,Q\}
    \Big\}
\end{align*}
and let
\begin{align*}
    \mathcal{X}^{\beta}_{q}
    =
    \Big\{
    x\in \mathcal{X}_{q}^{reg}:
    &\ \norm{x}_{\infty}< B, \\
    &\ x_{1q}>x_{3q}+\frac{1}{2}, 
       \quad x_{3q}>x_{2q}+\frac{1}{2}, \\
    &\ x_{2r}>x_{3r}+\frac{1}{2}, 
       \quad x_{3r}>x_{1r}+\frac{1}{2},
       \quad r\in\{1,\ldots,Q\}\setminus\{q\}
    \Big\}.
\end{align*}
where we use the shorthand $x_{ts,q}=x_{t,q}-x_{s,q}$ for conciseness.
The set $\mathcal{X}^{\beta}_{q}$ can be thought of as identifying the common slope coefficient $\beta$, while the sets $\mathcal{X}^{\gamma}_{q,\epsilon}$ pin down the state-dependence parameters $\gamma_q$. These bounded covariate regions are constructed in such a way that the Jacobian of the identifying map, $\nabla\Xi(\theta)$, possesses a particular sign and dominance structure.  Specifically, for $\Xi(\theta)$ defined in \eqref{EmpiricaliIdentifyingMoments}, they imply that $\nabla\Xi(\theta)$ has all entries positive and is strictly diagonally dominant, uniformly over $\theta\in\Theta$, once $\epsilon_1,\ldots,\epsilon_Q$ are chosen sufficiently small. The latter condition means that the relevant own-derivative terms dominate the cross-derivative terms in absolute value. Appendix Lemma \ref{lemma:pmatrix} shows that these properties imply that $\nabla\Xi(\theta)$ is a $P$-matrix, as desired.

\begin{proposition} \label{prop:Identification_Habit}
    Suppose that Assumptions \ref{ass:Compact_Theta}-\ref{ass:Identification_Restrictions} hold and that
    \begin{align} 
    &P(Y_{0}=y_0,X\in \mathcal{X}^{\beta}_{y_{0}})>0 \quad \text{for all } y_{0}\in\{1,\ldots,Q\} \label{ass:Well_Defined_Sets_beta}  \\
    &P(Y_{0}=y_0,X\in  \mathcal{X}^{\gamma}_{y_{0},\epsilon})>0 \quad \text{for all } y_{0}\in\{1,\ldots,Q\} \text{ and } \epsilon\in (0,3B) \label{ass:Well_Defined_Sets_gamma}
    \end{align}
    Then, there exists $\epsilon_{1},\ldots,\epsilon_{Q}$, sufficiently small positive constants, such that the mapping 
    \begin{align} \label{EmpiricaliIdentifyingMoments}
    \Xi(\theta) = \begin{pmatrix}
       \mathbb{E}\left[m_{y_{0}}^{\sigma_{1},1}(Y,X,\theta)|Y_{0}=1,X\in \mathcal{X}^{\beta}_{1}\right] \\
\mathbb{E}\left[m_{y_{0}}^{\sigma_{1},1}(Y,X,\theta)|Y_{0}=1,X\in \mathcal{X}^{\gamma}_{1,\epsilon_{1}}\right] \\
\mathbb{E}\left[m_{y_{0}}^{\sigma_{1},2}(Y,X,\theta)|Y_{0}=2,X\in \mathcal{X}^{\gamma}_{2,\epsilon_{2}}\right] \\
\vdots \\
\mathbb{E}\left[m_{y_{0}}^{\sigma_{1},Q}(Y,X,\theta)|Y_{0}=Q,X\in \mathcal{X}^{\gamma}_{Q,\epsilon_{Q}}\right]
    \end{pmatrix}
\end{align}
is injective on $\Theta$. Thus, $\theta_{0}$ is the unique solution to $\Xi(\theta)=0$, i.e., it is identified. 
\end{proposition} 
The proof is given in the Appendix. The use of the initial condition $Y_0=1$ in the first component of $\Xi(\theta)$, which identifies $\beta$, is inessential and could be replaced by any other initial condition.\footnote{Accordingly, \eqref{ass:Well_Defined_Sets_beta} can be relaxed to there exists $y_0\in\{1,\ldots,Q\}$ such that $P(Y_0=y_0,X\in\mathcal{X}^{\beta}_{y_0})>0$, without modification to the proof.} By contrast, our proof recovers the state-dependence parameters by exploiting observations from each initial condition, which explains the support requirement in \eqref{ass:Well_Defined_Sets_gamma}. \\
\indent In sum, Proposition \ref{prop:Identification_Habit} establishes point identification of the parameter $\theta_0$ in the empirical specification \eqref{EmpiricalSpec} through the moment map $\Xi(\theta)$. Because these moments are evaluated over events with positive probability, the argument also shows that $\theta_0$ is, in principle, $\sqrt{n}$-estimable. The particular map in \eqref{EmpiricaliIdentifyingMoments}, however, is best viewed as a constructive device for proving identification. Using it directly for estimation would require selecting the tuning constants $\epsilon_1,\ldots,\epsilon_Q$ and restricting attention to small covariate regions, which would be cumbersome and generally inefficient. %
For this reason, our empirical application instead uses a tuning-parameter-free and more flexible GMM procedure based on the collection of moment functions in \eqref{MainMoment}.

\section{Empirical application}

\label{sec:empirical}

In this section, we illustrate the use of the moment conditions derived
above for estimation of a simple model for product choice as a function of
prices in the presence of state dependence and persistent unobserved
heterogeneity. The product that we consider is soft drinks,
and the parameters of interest are a common coefficient on price in the indirect utility function and product-specific coefficients on an indicator for whether the same product was chosen on the previous purchase occasion.

The data are  from 2019 and come from NielsenIQ.\footnote{The results presented here are based on our own analyses calculated (or derived) based in part on data from Nielsen Consumer LLC and marketing databases provided through the NielsenIQ Datasets at the Kilts Center for Marketing Data Center at The University of Chicago Booth School of Business. 
The conclusions drawn from the NielsenIQ data are those of the researchers and do not reflect the views of NielsenIQ. NielsenIQ is not responsible for, had no role in, and was not involved in analyzing and preparing the results reported herein.
} There are 8,100
households making a total of 47,933 purchases of soft drinks. The categories are aggregated into Cola, Root Beer/Ginger Flavor, Citrus Flavor, Soft Drink with Unique Flavor, and Other. The households are located in one of the ten largest Nielsen DMAs.
For each of the categories, we calculate the price as a weighted average of sales in all stores of the household  Nielsen DMA within the week of purchase.
In addition to household ID, option chosen and prices, we also have information on the Nielsen DMA, household income, household size and
race.\footnote{A Designated Market Area (DMA) is a federally registered trademark and service mark of The Nielsen Company (US), LLC. Income is measured in intervals on the real line with cutoff at \$5,000, \$8,000, \$10,000, \$15,000, \$20,000, \$25,000, \$30,000, \$35,000, \$40,000, \$45,000, \$50,000, \$60,000, and \$100,000. We record income as 2.5 for observations in the bottom interval and as 125 for observations in the top interval; for the remaining observations income is recorded at the midpoint of the interval divided by 1000.} Table \ref{tab:summary} displays summary statistics for the data.

\begin{table}[h]
\caption{Summary Statistics}
\vspace{-0.2in}
\label{tab:summary}
\begin{center}
\begin{tabular}{lccc}
\toprule\toprule
 & No. Obs. & Mean & Std. Dev. \\
\midrule
Price (Cola)                    &      47,933 &      3.190 &      0.628 \\
Price (Root Beer/Ginger)        &      47,933 &      3.261 &      0.664 \\
Price (Citrus)                  &      47,933 &      3.308 &      0.671 \\
Price (Unique Flavor)            &      47,933 &      3.348 &      0.944 \\
Price (Other)                   &      47,933 &      4.491 &      1.879 \\
Household size                  &      47,933 &      2.671 &      1.384 \\
Income (in \$1,000)             &      47,933 &     72.940 &     37.644 \\
White                           &      47,933 &      0.715 &      0.451 \\
Black                           &      47,933 &      0.173 &      0.379 \\
Asian                           &      47,933 &      0.042 &      0.200 \\
Other                           &      47,933 &      0.069 &      0.254 \\
\bottomrule\bottomrule
\end{tabular}
\end{center}
\vspace{-0.3in}
{\singlespacing \footnotesize The table displays summary statistics for the variables used here. The data comes  from NielsenIQ, and covers 47,933 purchases in the ten largest Nielsen DMAs in 2019. A unit of observation is a purchase occasion.
Prices are measured in cents per ounce.}
\end{table}

Figure \ref{fig:purchases_hist} shows the histogram of the
number of purchases per household and Table \ref{tab:transitions} shows the number of transitions between products. 

\begin{figure}[h]
\caption{Histogram of Purchases per Household}
\vspace{-0.2in}
\label{fig:purchases_hist}
\begin{center}
\includegraphics[width =
0.6\textwidth]{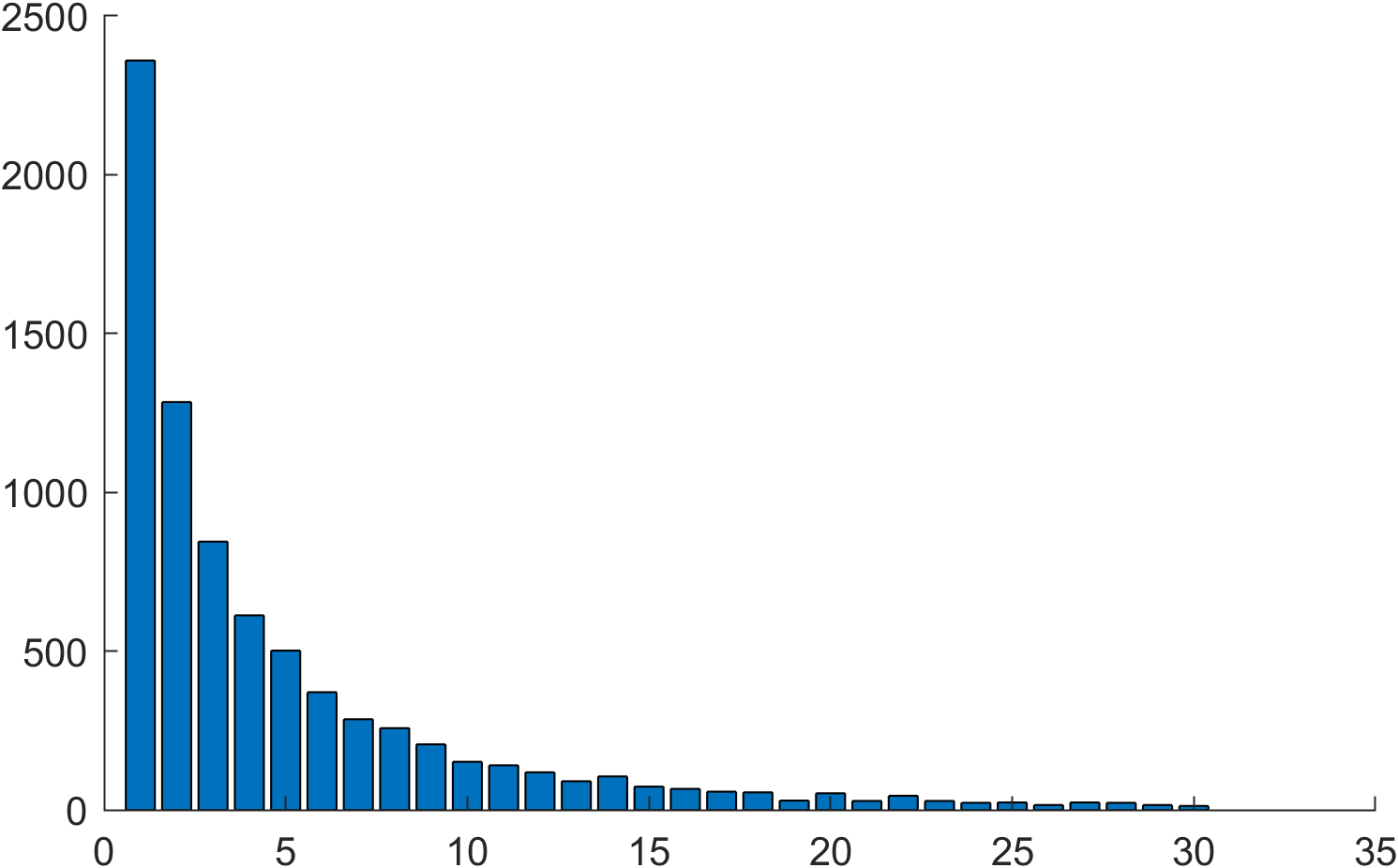}
\end{center}
\vspace{-0.3in}
{\singlespacing \footnotesize The figure shows the histogram for the number of soft drink purchases for 8,100 households in the ten largest Nielsen DMAs in 2019. 186 households with more than 30
purchases are not displayed. The data comes from NielsenIQ.}
\end{figure}

\begin{table}[h]
\caption{Transitions Between Products}
\vspace{-0.2in}
\label{tab:transitions}
\begin{center}
\begin{tabular}{lccccc}
\toprule\toprule
& Cola & Root Beer/Ginger & Citrus & Unique Flavor & Other \\ 
\midrule
Cola & {8528} & 1142 & 1179 & 600 & 1435 \\ 
Root Beer/Ginger & 1137 & {3911} & 639 & 260 & 1285 \\ 
Citrus & 1142 & 620 & {2871} & 351 & 825 \\ 
Unique Flavor & 603 & 261 & 360 & {1836} & 355 \\ 
Other & 1468 & 1255 & 851 & 371 & {6548}\\
\bottomrule\bottomrule
\end{tabular}%
\end{center}
\vspace{-0.3in}
{\singlespacing \footnotesize The table shows the number of times a purchase of a type of soft drink (listed in the column to the left) is followed by a purchase of a possibly different type (listed in the row on top). The data comes  from NielsenIQ, and covers 47,933 purchases in the ten largest Nielsen DMAs in 2019.
}
\end{table}

Table \ref{tab:transitions} shows substantial persistence in choices, with 59\% of purchases repeating the previous category. However, persistence varies considerably by product. For example, 66\% of cola purchases are followed by another cola purchase, while only 49\% of citrus drink purchases are followed by citrus drinks. The large degree of persistence suggests that true state dependence and/or persistent unobserved preferences play an important role.
These patterns motivate a model that allows for product-specific state dependence as well as unrestricted individual-product heterogeneity.  
Specifically, we consider the linear specification \eqref{EmpiricalSpec} with price as an explanatory variable. That is,
\begin{equation*}
z_{t}(q,y_{t-1},\text{Price},\theta)=\text{Price}_{t,q}\times \beta+
\gamma_{q}\mathbbm{1}\left\{y_{t-1}=q\right\}
\end{equation*}
where $\theta=(\beta^{\prime },\gamma_{1},\ldots,\gamma_{Q})^{\prime }$. In
this model, the $\gamma_{q}$'s measure the product specific state dependence
and the fixed effects $A_{qi}$ control unrestricted individual-product
specific heterogeneity.  Proposition \ref{prop:Identification_Habit} in the previous section established that $\theta$ is identified under appropriate support conditions.

Let $Z_i^{s:t}=(Z_{is},Z_{i,s+1},\ldots,Z_{it})$ for $s<t$, and let $T_i$ denote the number of observed purchase occasions for individual $i$ after the initial condition. We estimate the model by generalized method of moments (GMM), using as building blocks the three-period conditional moment functions displayed in \eqref{moment_exp}. For every individual $i$ with $T_i\geq 3$, we exploit the moment functions associated with each consecutive three-period window. That is, for each soda category $q$, the collection of conditional moment functions we consider is
\begin{align*}
    m_{Y_{i0}}^{\sigma_{1},\sigma_{+},q}
\left(Y_i^{1:3},\mathrm{Price}_i^{1:3},\theta\right),\quad
m_{Y_{i1}}^{\sigma_{1},\sigma_{+},q}
\left(Y_i^{2:4},\mathrm{Price}_i^{2:4},\theta\right),
\quad \ldots,\quad
m_{Y_{i,T_i-3}}^{\sigma_{1},\sigma_{+},q}
\left(Y_i^{T_i-2:T_i},\mathrm{Price}_i^{T_i-2:T_i},\theta\right).
\end{align*}

Following \citet{honore2025moment}, we first pre-normalize each conditional moment function by dividing it by the sum of the absolute values of the distinct summands appearing in its expression.\footnote{Appendix~\ref{app:identification:proof} gives a disaggregated expression of $m_{Y_0}^{\sigma_{1},\sigma_{+},q}(Y,X,\theta)$, making explicit its constituent terms.} This normalization uniformly bounds the conditional moment functions and their gradients over the parameter space. As in the binary choice setting considered by \citet{honore2025moment}, we find that this rescaling improves the finite-sample performance of the resulting GMM estimators. 

Next, we form unconditional moment functions by interacting the rescaled conditional moment functions $\widetilde m_{Y_{i,s-1}}^{\sigma_{1},\sigma_{+},q}
\left(Y_i^{s:s+2},\mathrm{Price}_i^{s:s+2},\theta\right)$
with indicators for the initial condition $Y_{is-1}$ and the relative price differences $\mathrm{Price}_{is,q'}-\mathrm{Price}_{is,\mathrm{cola}}$, $\mathrm{Price}_{is+1,q'}-\mathrm{Price}_{is+1,\mathrm{cola}}$, and $\mathrm{Price}_{is+2,q'}-\mathrm{Price}_{is+2,\mathrm{cola}}$ for each product category $q'$ other than cola.\footnote{Because only price differences matter in the multinomial logit model, we use prices relative to cola rather than price levels.}

Finally, to keep the dimension of the estimating moment vector manageable, we sum up these unconditional moments across all three-period windows within each soda category $q$, for each individual $i$. This procedure ultimately yields a total of 85 moment conditions. Since this number is large relative to the sample size,\footnote{The weighting matrix has 3,655 unique elements, while the sample size is 36,859.} estimating a fully unrestricted optimal weighting matrix may be unstable. We therefore use a two-step GMM procedure in which the second-step weighting matrix is diagonal, with each diagonal element equal to the inverse of the sample variance of the corresponding moment function evaluated at the initial GMM estimator.

The calculations in Section \ref{sec: Linear Panel Data Models with Normal Errors} above reveal that moment conditions derived from \eqref{SExplicitGeneric} can have heavy--tailed distributions. The results presented in Table \ref{tab:moment_stats_unrescaled_nonzero} suggest that this is not an issue for the moments used in the estimation in this section. Specifically,  Table \ref{tab:moment_stats_unrescaled_nonzero} presents summary statistics for the empirical distribution of the 3-period conditional moment functions in equation \eqref{moment_exp} evaluated at the GMM estimator. Since each of the moments has a high probability of being 0, we present that probability along with summary statistics for the moment functions conditional on not being 0. Conditional on not being 0, all of the five moment functions have skewness close to 0 and the kurtosis is between $1.24$ and $1.35$. Consistent with equation \eqref{moment_exp}, the minimum value is always $-1$, while the maximum value is below $2.5$.

\begin{table}[!ht]
\centering
\caption{Summary Statistics for the Moment Functions}
\label{tab:moment_stats_unrescaled_nonzero}
\resizebox{\textwidth}{!}{%
\begin{tabular}{ccccccccccc}
\toprule\toprule
&  &  & \multicolumn{8}{c}{Summary Statistics Conditional on $ m^{\sigma_{1},\sigma_{+},q}\neq 0$}
\\ \cline{4-11}
$q$ & \multicolumn{1}{|c}{$\Pr (m^{\sigma_{1},\sigma_{+},q}=0)$} &  & Mean & Variance & SD & 
Skewness & Kurtosis & Min & Median & Max \\ \hline
1 & \multicolumn{1}{|c}{$0.705$} &  & \multicolumn{1}{r}{$0.010$} & 
\multicolumn{1}{r}{$0.846$} & \multicolumn{1}{r}{$0.920$} & 
\multicolumn{1}{r}{$0.061$} & \multicolumn{1}{r}{$1.348$} & 
\multicolumn{1}{r}{$-1.000$} & \multicolumn{1}{r}{$-0.003$} & 
\multicolumn{1}{r}{$2.243$} \\ 
2 & \multicolumn{1}{|c}{$0.791$} &  & \multicolumn{1}{r}{$-0.009$} & 
\multicolumn{1}{r}{$0.865$} & \multicolumn{1}{r}{$0.930$} & 
\multicolumn{1}{r}{$0.073$} & \multicolumn{1}{r}{$1.284$} & 
\multicolumn{1}{r}{$-1.000$} & \multicolumn{1}{r}{$0.001$} & 
\multicolumn{1}{r}{$2.415$} \\ 
3 & \multicolumn{1}{|c}{$0.811$} &  & \multicolumn{1}{r}{$-0.013$} & 
\multicolumn{1}{r}{$0.892$} & \multicolumn{1}{r}{$0.944$} & 
\multicolumn{1}{r}{$0.080$} & \multicolumn{1}{r}{$1.241$} & 
\multicolumn{1}{r}{$-1.000$} & \multicolumn{1}{r}{$-0.010$} & 
\multicolumn{1}{r}{$2.320$} \\ 
4 & \multicolumn{1}{|c}{$0.906$} &  & \multicolumn{1}{r}{$0.025$} & 
\multicolumn{1}{r}{$0.947$} & \multicolumn{1}{r}{$0.973$} & 
\multicolumn{1}{r}{$0.063$} & \multicolumn{1}{r}{$1.243$} & 
\multicolumn{1}{r}{$-1.000$} & \multicolumn{1}{r}{$0.001$} & 
\multicolumn{1}{r}{$1.856$} \\ 
5 & \multicolumn{1}{|c}{$0.755$} &  & \multicolumn{1}{r}{$-0.014$} & 
\multicolumn{1}{r}{$0.854$} & \multicolumn{1}{r}{$0.924$} & 
\multicolumn{1}{r}{$0.088$} & \multicolumn{1}{r}{$1.313$} & 
\multicolumn{1}{r}{$-1.000$} & \multicolumn{1}{r}{$-0.008$} & 
\multicolumn{1}{r}{$2.235$} \\ 
\bottomrule\bottomrule
\end{tabular}
}
\end{table}

In addition to the GMM estimator based on the moments derived in this paper, we also calculate the
maximum likelihood estimator of a multinomial logit models that uses prices, past choice, household size, household income, and race dummies as explanatory variables. 
The GMM estimation requires at least four purchase occasions for a household, while the multinomial logit maximum likelihood estimator only requires two. We therefore calculate maximum likelihood estimator on the sample of household with at least two purchases as well as for the data used for the GMM estimator.  The results are reported in Table \ref{TABLE Estimation Results}. 

The maximum likelihood estimates that do not account for unobserved heterogeneity imply substantial state dependence, with coefficients between 1.5 and 2.5 depending on the product category. In contrast, the GMM estimates that control for fixed effects are dramatically smaller, ranging from 0.02 to 0.29.
This difference suggests that a large share of the observed persistence in choices is due to persistent unobserved heterogeneity rather than true state dependence. After controlling for this heterogeneity, state dependence remains economically modest and is close to zero for some categories, such as citrus and other soft drinks. The large difference between the MLE and GMM estimates highlights the importance of controlling for unobserved heterogeneity when estimating state dependence in discrete choice models.

In contrast, the estimated price coefficient is remarkably stable across specifications, at approximately $-0.11$, suggesting that ignoring unobserved heterogeneity primarily biases estimates of state dependence rather than price sensitivity. At a superficial level, this is not surprising since the prices are plausibly independent of the household  fixed effects.

\begin{table}[h]
\caption{Estimation Results: State Dependence and Price Sensitivity}
\label{TABLE Estimation Results}
\begin{center}
\input{./Tables/table_Nielsen.tex}
\end{center}
\vspace{-0.3in}
{\singlespacing \footnotesize
The table shows the results from maximum likelihood estimation and GMM estimation of a dynamic multinomial logit model. The
data comes from NielsenIQ and covers 47,933 purchases in the ten largest
Nielsen DMAs in 2019.
The models estimated by maximum likelihood also include the time invariant variables household income, household size and race dummies for White, Black and Asian as explanatory variables. 
The GMM specification allows for unrestricted fixed effects. Asymptotic standard errors are reported in parentheses.}
\end{table}

\section{Conclusion}

This paper has shown how to produce moment conditions free of the fixed effects
for the dynamic multinomial logit model. %
Our approach is based on a simple observation: permuting outcomes across time periods leaves the
contribution of the alternative-specific fixed effects unchanged, so the
ratio of a transformed-model probability which permutes the order of outcomes to the true-model probability is free of those effects. When the chosen permutation is cyclic, the transformed model defines
a proper distribution. This alone yields a valid moment condition,  and contrasting two such models delivers additional ones for the common parameters. The construction requires
at least three time periods after the initial condition, and we provide an
explicit formula for the number of linearly independent moments that are generated this way.
This number grows quickly in both the number of alternatives and the number of time
periods.

Two features of the construction of moment conditions in this paper are worth emphasizing. First, it is explicit. 
Unlike the abstract characterization obtained from functional differencing, our approach
delivers moment functions in closed form that can be used directly for
estimation. Second, the construction is general. The only model-specific ingredient is that
the static counterpart of the model admits a sufficient statistic for the fixed
effects that is symmetric in the time indices. Whenever such a sufficient statistic exists,
the transformed-model device can be deployed to generate fixed effects-free moment conditions for the
corresponding dynamic model with lagged dependent variables. This unifies the
multinomial logit with other results, and provides one way to proceed in
panel data settings that have not yet been studied.

We put our moment conditions to use in two ways. We establish point
identification of the parameters in our empirical specification under support
conditions that, unlike conditional-likelihood arguments, do not require
covariates to be constant across periods. 
We also leverage the moment conditions to estimate a multinomial logit model with fixed effects in an  application to household soft drink
purchases. Accounting for unobserved heterogeneity in this application sharply reduces the estimated
degree of state dependence relative to a model that ignores it, while leaving the
price coefficient essentially unchanged. This suggests that  much of the persistence
in observed choices can be explained by stable unobserved preferences rather than genuine
state dependence.

Several questions remain open. For example, our count of linearly independent moment
conditions is supported by a proof for the binary case and otherwise for three time
periods. We verified the validity of the claim numerically more broadly, but a general analytical proof
for more than two alternatives is still lacking. Also, while the
transformed-model moments are always valid, they need not be well behaved when applied in specific settings. As the 
linear panel data Gaussian example illustrates, they can be 
heavy--tailed in which case they may be best viewed
as a constructive existence result that complements, rather than replaces,
model-specific reasoning about how best to eliminate the fixed effects.
Characterizing efficient or well-behaved combinations of the available moments,
and extending the construction to richer forms of dependence and to settings
without independence of irrelevant alternatives, are natural directions for
future work.

\label{sec:conclusion}

\section*{Acknowledgement}
This paper benefited from the use of generative AI tools to assist with language editing and \LaTeX{} formatting; all output was carefully reviewed by the authors. All substantive content, results, and any remaining errors are the authors' responsibility.

\setlength{\bibsep}{2pt}

\newpage
\appendix

\section{Discussion of Conjecture~\ref{conj:counting}}
\label{app:counting}

This appendix discusses two complementary strategies for proving Conjecture~\ref{conj:counting}. The first uses the polynomial structure of the model probabilities in the fixed effect parameter to derive a lower bound on the total number of linearly independent moment conditions, which together with a known matching upper bound pins down the count exactly for $Q = 2$. The second argues directly about linear independence of the $\ell(Q, T)$ moments in \eqref{MomentParameterized}, which we establish rigorously for $T = 3$ and any $Q \geq 2$.

\subsection{Strategy 1: polynomial expansion of the likelihood}
\label{app:counting:poly}

Following \citet{dobronyi2021identification}, define $\eta_q := \exp(\alpha_q)$ for $q = 1, \ldots, Q$, and $\pi_t(q, r) := \exp(z_t(q, r, x, \theta))$. Without loss of generality, normalize $\eta_1 = 1$.\footnote{This is without loss of generality since $p(y \mid y_0, x, \theta, \alpha) = p(y \mid y_0, x, \theta, \alpha^*)$ for $\alpha^* = (0, \alpha_2 - \alpha_1, \ldots, \alpha_Q - \alpha_1)$.} The joint probability \eqref{JointProb} factors as
\begin{align}
p(y \mid y_0, x, \theta, \alpha) \;=\; \kappa(\eta) \cdot \tilde{p}(y, \eta),
\label{PolyFactor}
\end{align}
where $\kappa(\eta) = \Big[ \big(\sum_{r=1}^Q \pi_1(r, y_0) \eta_r\big) \prod_{t=2}^T \prod_{l=1}^Q \big(\sum_{r=1}^Q \pi_t(r, l) \eta_r\big) \Big]^{-1}$ depends on $\eta$ but not on $y$, and
\begin{align}
\tilde{p}(y, \eta) \;=\; \pi_1(y_1, y_0) \, \eta_{y_1} \prod_{t=2}^T \bigg\{ \pi_t(y_t, y_{t-1}) \, \eta_{y_t} \prod_{l=1}^Q \left[\sum_{r=1}^Q \pi_t(r, l) \, \eta_r\right]^{1 - \mathbbm{1}\{y_{t-1} = l\}} \bigg\}
\label{PolyP}
\end{align}
is a polynomial in $\eta = (\eta_2, \ldots, \eta_Q)$ of total degree at most $(T-1)(Q-1) + T = TQ - Q + 1$.

Writing the polynomial in monomial form, $\tilde{p}(y, \eta) = \sum_k c_k(y) \prod_{i=2}^Q \eta_i^{k_i}$, the moment condition \eqref{MomentConditionSum} (which must hold for all $\eta \in \mathbb{R}_{+}^{Q-1}$) is equivalent to the system of linear constraints
\begin{align}
\sum_{y} m(y) \, c_k(y) \;=\; 0 \quad \text{for every multi-index } k \in \mathbb{N}^{Q-1} \text{ with } |k| \leq TQ - Q + 1.
\label{PolyConstraint}
\end{align}
The total number of linearly independent moment conditions is therefore $Q^T - \mathrm{rank}\{c_k(\cdot) : k\}$. Since the number of multi-indices $k$ with $|k| \leq TQ - Q + 1$ in $Q - 1$ variables is $\binom{TQ}{Q-1}$, this rank is at most $\binom{TQ}{Q-1}$, yielding the lower bound
\begin{align}
\text{(\# linearly independent moment conditions)} \;\geq\; Q^T - \binom{TQ}{Q-1}.
\label{WeakBound}
\end{align}

\paragraph{The case $Q = 2$.}
For $Q = 2$, $\tilde{p}(y, \eta_2)$ is a univariate polynomial of degree at most $2T - 1$, so the constraints \eqref{PolyConstraint} involve exactly $2T$ coefficient vectors $c_0, c_1, \ldots, c_{2T-1} \in \mathbb{R}^{2^T}$. The bound \eqref{WeakBound} gives at least $2^T - 2T = \ell(2, T)$ linearly independent moment conditions. \citet{honore2025moment} construct $\ell(2, T)$ explicit linearly independent moment functions, which gives the same lower bound. The matching upper bound is established in subsequent work; see \citet{dobronyi2021identification} and \citet{dano2023transition}. This proves that the number of linearly independent moment conditions equals $\ell(2, T)$ for $Q = 2$ and any $T \geq 3$.

\paragraph{The case $Q \geq 3$.}
For $Q \geq 3$, the bound \eqref{WeakBound} is not sharp and can even be negative. For example, $\binom{9}{2} = 36 > 27 = 3^3$, so \eqref{WeakBound} is uninformative at $(Q, T) = (3, 3)$. The reason is that the $\binom{TQ}{Q-1}$ coefficient vectors $c_k(\cdot) \in \mathbb{R}^{Q^T}$ are not all linearly independent: the product structure of $\tilde{p}(y, \eta)$ as a product of linear forms in $\eta$ forces algebraic relations among the $c_k$. Identifying these relations in closed form, and thereby sharpening \eqref{WeakBound}, is an open problem.

\subsection{Strategy 2: linear independence of the moments in \eqref{MomentParameterized}}
\label{app:counting:construction}

The second strategy works directly with the $\ell(Q, T)$ explicit moment functions in \eqref{MomentParameterized}. If these moments are linearly independent for generic $(y_0, x, \theta)$, then at least $\ell(Q, T)$ linearly independent moment conditions exist. Combined with a matching upper bound on that number, which is known for $Q=2$ (see Strategy~1), this establishes Conjecture~\ref{conj:counting}.

We prove linear independence rigorously for $T = 3$ and any $Q \geq 2$. For $T = 3$ the only allowed value is $j = 0$, which forces $\mathbf{t} = \emptyset$, $s_1 = 1$, $s_2 = 2$, and varies $q_0 \in \{1, \ldots, Q\}$. The two cyclic permutations of $\{1, 2, 3\}$ have generation orders $\tau_1 = (1, 2, 3)$ (so that $\sigma_1$ is the cyclic increment $\sigma_+$ and $S^{\sigma_1} = 1$) and $\tau_2 = (2, 1, 3)$ (the reverse cyclic). The parameterization \eqref{MomentParameterized} assigns $s_1$ to $\sigma_1$, so the two labels are interchanged here relative to the $T=3$ example of Section~\ref{Sec: Moment construction}, where $\sigma_1$ denotes the reverse cycle. The resulting $Q$ moment functions are
\begin{align}
m_q(y) \;:=\; \mathbbm{1}\{y_1 = q\} - \mathbbm{1}\{y_2 = q\} \, S^{\sigma_2}_{y_0}(y, x, \theta), \qquad q = 1, \ldots, Q.
\label{MomentT3App}
\end{align}

\begin{proposition}
\label{prop:T3indep}
For $T = 3$ and any $Q \geq 2$, the moment functions $m_1, \ldots, m_Q$ in \eqref{MomentT3App} are linearly independent in $\mathbb{R}^{Q^T}$ for generic $(y_0, x, \theta_0)$.
\end{proposition}

\begin{proof}
For each $q \in \{1, \ldots, Q\}$, fix any $c_q \neq q$ and define the evaluation point
\[
y_q \;:=\; (q, q, c_q) \in \{1, \ldots, Q\}^3.
\]
At $y_q$, both indicators in \eqref{MomentT3App} reduce to $\mathbbm{1}\{q' = q\}$, so for each $q' \in \{1, \ldots, Q\}$,
\begin{align*}
m_{q'}(y_q) \;=\; \mathbbm{1}\{q = q'\} \, \big[1 - S^{\sigma_2}_{y_0}(y_q, x, \theta)\big].
\end{align*}
The $Q \times Q$ matrix $\big[m_{q'}(y_q)\big]_{q', q}$ is therefore diagonal with diagonal entries $1 - S^{\sigma_2}_{y_0}(y_q, x, \theta)$.

It remains to verify that $S^{\sigma_2}_{y_0}(y_q, x, \theta) \neq 1$ for generic $(y_0, x, \theta_0)$. The reverse cyclic $\sigma_2$ acts as $\sigma_2(1) = 3$, $\sigma_2(2) = 1$, $\sigma_2(3) = 2$, so $\sigma_2(0) := \sigma_2(3) = 2$. Using \eqref{DefS},
\begin{align*}
S^{\sigma_2}_{y_0}((q, q, c_q), x, \theta) \;=\; \exp\!\Big[ &\big(z_1(q, y_0) - z_1(q, y_0)\big) + \big(z_2(c_q, q) - z_2(q, q)\big) \\
& + \big(z_3(q, q) - z_3(c_q, q)\big) \Big] \\
=\; \exp\!\Big[ &\big(z_2(c_q, q) - z_3(c_q, q)\big) - \big(z_2(q, q) - z_3(q, q)\big) \Big].
\end{align*}
For the linear specification \eqref{LinearSpec}, $z_t(r, l, x, \theta) = \gamma_{r, l} + x_t' \beta_r$, so the $\gamma$ terms cancel and
\[
S^{\sigma_2}_{y_0}(y_q, x, \theta) \;=\; \exp\!\big[(x_2 - x_3)' (\beta_{c_q} - \beta_q)\big].
\]
For generic $x$ with $x_2 \neq x_3$ and generic $\beta$ with $\beta_{c_q} \neq \beta_q$ for each $q$, the exponent is non-zero and $S^{\sigma_2}_{y_0}(y_q, x, \theta) \neq 1$. The diagonal matrix $\big[m_{q'}(y_q)\big]$ is therefore invertible, and the moment functions $m_1, \ldots, m_Q$ are linearly independent.
\end{proof}

Proposition~\ref{prop:T3indep} establishes the first half of Conjecture~\ref{conj:counting} (linear independence of the constructed moments) for $T = 3$ and any $Q \geq 2$, and implies that at least $\ell(Q, 3) = Q$ linearly independent moment conditions of the form \eqref{MomentCondition} exist for generic $(y_0, x, \theta_0)$. The second half of the conjecture (that every valid moment function is a linear combination of the constructed moments, equivalently that there are no more than $\ell(Q, 3)$ linearly independent moment conditions) is supported by direct numerical computation for $T = 3$ but does not have a clean analytical proof for $Q \geq 3$.

\paragraph{The obstacle for $T \geq 4$.}
For $T \geq 4$, linear independence of the $\ell(Q, T)$ moment functions in \eqref{MomentParameterized} has been verified by direct numerical computation for all values of $(Q, T)$ in Table~\ref{tab:counting}, but a general analytical proof has so far eluded us. A natural extension of the $T = 3$ argument groups the moments hierarchically by the number of shared time indices $j$ and, within each level, by the shared-position values $\mathbf{t}$ and $\mathbf{q}$. Within a fixed group the moments reduce to a structure analogous to the $T = 3$ case, and a similar diagonal evaluation-point argument applies on a restricted slice of the outcome space. Combining the slices into a single linear-independence statement requires bookkeeping that we have not been able to carry out cleanly for arbitrary $(Q, T)$.

\subsection{Summary}

Combining the two strategies, the count $\ell(Q, T)$ in Conjecture~\ref{conj:counting} is established analytically for $Q = 2$ and any $T \geq 3$ (by Strategy~1) and, as a lower bound, for $T = 3$ and any $Q \geq 2$ (by Strategy~2). For $Q \geq 3$ the conjecture remains open in both halves: at $T = 3$ we have the lower bound but not exhaustiveness, and for $T \geq 4$ neither half is established. Equality has been verified numerically for all $(Q, T)$ in Table~\ref{tab:counting} but not proved analytically. We leave a complete analytical proof of Conjecture~\ref{conj:counting} for future work.

\section{Generalized moment conditions for non-cyclic permutations}
\label{app:generalized}

Theorem~\ref{th:MainTheorem} constructs moment conditions from pairs of cyclic permutations $\sigma_1,\sigma_2$. In this appendix, we show that the construction extends to pairs of general (not necessarily cyclic) permutations, provided they satisfy a compatibility condition on their cycle structures. In all cases we have checked numerically, the resulting moment conditions do not generate any linearly independent moments beyond those already obtained from cyclic permutations in Theorem~\ref{th:MainTheorem}, so this generalization is not needed for Conjecture~\ref{conj:counting}. However, the result is of independent interest, as it clarifies the structure of the transformed model approach and may be useful for extensions to other models.

\subsection{Transformed models for non-cyclic permutations}

For a general (not necessarily cyclic) permutation $\sigma:\{1,\ldots,T\}\to\{1,\ldots,T\}$, the transformed model $p^\sigma$ defined in \eqref{AuxModel} does not correspond to a single Markov chain through all $T$ outcomes. Instead, $p^\sigma$ factorizes as a product of independent Markov chains, one for each cycle of $\sigma$.

To see this, recall that every permutation $\sigma$ can be decomposed into disjoint cycles. Let $\sigma = c_1 \circ c_2 \circ \cdots \circ c_L$ be the cycle decomposition, with $L \geq 1$ cycles. Each cycle $c_\ell$ acts on a subset $C_\ell \subset \{1,\ldots,T\}$, and $\{C_1,\ldots,C_L\}$ is a partition of $\{1,\ldots,T\}$. Let $C_*$ denote the cycle containing $T$, and write $C_* = \{T, \sigma(T), \sigma^2(T), \ldots, \sigma^{|C_*|-1}(T)\}$. For the cycle $C_*$, define the generation order $\tau_*(k) := \sigma^k(T)$ for $k = 1,\ldots,|C_*|$, so that $\tau_*(|C_*|) = T$. Similarly, for each cycle $C_\ell$ not containing $T$, pick any element $a_\ell \in C_\ell$ and define the generation order within that cycle by $\tau_\ell(k) := \sigma^k(a_\ell)$ for $k = 1,\ldots,|C_\ell|$, with $\tau_\ell(|C_\ell|) = a_\ell$. In \eqref{AuxModelFactorized} below we also use the convention $\tau_\ell(0) := a_\ell$, which reflects that this chain returns to its starting outcome.

The transformed model $p^\sigma$ can then be written as
\begin{align}
p^\sigma(y \mid y_0, x, \theta, \alpha)
= \underbrace{p_1(y_{\tau_*(1)} \mid y_0, x, \theta, \alpha) \prod_{k=2}^{|C_*|} p_{\tau_*(k-1)+1}(y_{\tau_*(k)} \mid y_{\tau_*(k-1)}, x, \theta, \alpha)}_{\text{chain through cycle } C_*}
\label{AuxModelFactorized}
\end{align}
\[
\times \prod_{\ell : C_\ell \neq C_*}
\underbrace{\prod_{k=1}^{|C_\ell|} p_{\tau_\ell(k-1)+1}(y_{\tau_\ell(k)} \mid y_{\tau_\ell(k-1)}, x, \theta, \alpha)}_{\text{chain through cycle } C_\ell} .
\]
Each factor in \eqref{AuxModelFactorized} is a Markov chain that generates the outcomes $\{y_s : s \in C_\ell\}$ independently of the outcomes in other cycles. For cycles $C_\ell \neq C_*$, the chain starts from the outcome $y_{a_\ell}$, which also appears as the outcome generated in the last step of the chain (since $\tau_\ell(|C_\ell|) = a_\ell$). Hence, summing $p^\sigma$ over all $y_s$ with $s \in C_\ell$ for $C_\ell \neq C_*$ does \emph{not} yield one for generic $(x,\theta,\alpha)$, because these chains are ``closed loops'' that both start and end at $y_{a_\ell}$.

Nevertheless, Lemma~\ref{lemma:SnoAlpha} continues to hold: the ratio $S^\sigma_{y_0}(y,x,\theta) = p^\sigma(y \mid y_0,x,\theta,\alpha) / p(y \mid y_0,x,\theta,\alpha)$ is free of $\alpha$ for any permutation $\sigma$, not just cyclic ones. This is because the proof of Lemma~\ref{lemma:SnoAlpha} only uses the fact that $\sigma$ is a permutation (so that $\sum_t \alpha_{y_{\sigma(t-1)}} = \sum_t \alpha_{y_t}$), not that $\sigma$ is cyclic.

\subsection{Equivalence relation on permutations}

We define an equivalence relation $\sim$ on the set of all permutations of $\{1,\ldots,T\}$ as follows.

\begin{definition}
\label{def:equivrel}
Two permutations $\sigma_1,\sigma_2:\{1,\ldots,T\}\to\{1,\ldots,T\}$ satisfy $\sigma_1 \sim \sigma_2$ if and only if their cycle decompositions are identical except possibly for the cycles containing $T$.
\end{definition}

That is, $\sigma_1 \sim \sigma_2$ if every cycle that does \emph{not} contain $T$ appears identically in both $\sigma_1$ and $\sigma_2$. %
The cycles containing $T$ are allowed to differ but must have the same length. Note that $\sigma \sim \sigma$ always holds.

For example, with $T=7$, the permutations $\sigma_1 = [1,5,3][2,4,7][6]$ and $\sigma_2 = [1,5,3][2,7,4][6]$ satisfy $\sigma_1 \sim \sigma_2$, because the cycles $[1,5,3]$ and $[6]$ (which do not contain $T=7$) are identical, while the cycles $[2,4,7]$ and $[2,7,4]$ (which contain $T=7$) differ.

When $\sigma_1 \sim \sigma_2$, the two permutations share the same ``non-$T$ cycles.'' Denote their union by $C_{\mathrm{shared}} := \bigcup_{\ell: T \notin C_\ell} C_\ell$. Let $C_*^{(1)}$ and $C_*^{(2)}$ denote the cycles of $\sigma_1$ and $\sigma_2$ that contain $T$, respectively. Since $\sigma_1 \sim \sigma_2$, we have $C_*^{(1)} = C_*^{(2)}$ as sets (both equal $\{1,\ldots,T\} \setminus C_{\mathrm{shared}}$), though $\sigma_1$ and $\sigma_2$ may permute the elements of this set differently.

\subsection{Shared index set}

For two permutations $\sigma_1 \sim \sigma_2$, we now define the set of ``shared indices'' that determines the structure of the generalized moment conditions. Consider the generation orders within the cycles containing $T$:
\begin{align*}
\tau_*^{(1)}(k) := \sigma_1^k(T), \qquad \tau_*^{(2)}(k) := \sigma_2^k(T), \qquad k = 1,2,\ldots,|C_*|.
\end{align*}
Since the two cycles contain the same set of elements $C_*^{(1)} = C_*^{(2)}$, both generation orders are permutations of the same set, but they may list the elements in different orders. Additionally, the elements of $C_{\mathrm{shared}}$ provide further shared indices from the non-$T$ cycles.

We define the shared index set as
\begin{align}
{\cal I}_{\sigma_1,\sigma_2} := \big\{(t,t) : t \in C_{\mathrm{shared}}\big\} \cup  \big\{(t,t) : t = \sigma_1^{j}(T) = \sigma_2^{j}(T), \ j = 1,\ldots,k^* \big\} \cup \big\{(s_1,s_2)\big\},
\label{SharedIndexSet}
\end{align}
where the diagonal pairs $(t_k,t_k)$ correspond to all elements in the shared non-$T$ cycles together with the elements on which the two generation orders through $C_*$ still agree, and $(s_1,s_2)$ with $s_1 \neq s_2$ corresponds to the ``split point'' where the two cycles containing $T$ first diverge. More precisely, let $k^*$ be the largest integer such that $\tau_*^{(1)}$ and $\tau_*^{(2)}$ agree on their first $k^*$ entries:
\begin{align*}
 \sigma_1^{j}(T) = \sigma_2^{j}(T), \quad j = 1,\ldots,k^*,
\end{align*}
and $\sigma_1^{k^*+1}(T) \neq \sigma_2^{k^*+1}(T)$. Then $s_1 = \sigma_1^{k^*+1}(T)$ and $s_2 = \sigma_2^{k^*+1}(T)$, and the $t_k$ from the $T$-cycle are the shared elements $\sigma_1^{j}(T) = \sigma_2^{j}(T)$ for $j=1,\ldots,k^*$. Both generation orders start from $y_0$, so it is the shared initial segment of $\tau_*^{(1)}$ and $\tau_*^{(2)}$ that matters here, exactly as in Theorem~\ref{th:MainTheorem}.

We write $[{\sigma_1,\sigma_2}] := |{\cal I}_{\sigma_1,\sigma_2}|$ for the total number of shared index pairs.

\subsection{Generalized theorem}

\begin{theorem}
\label{th:GeneralizedMoments}
Let $\sigma_1, \sigma_2:\{1,\ldots,T\}\to\{1,\ldots,T\}$ be two permutations satisfying $\sigma_1 \sim \sigma_2$ (Definition~\ref{def:equivrel}), with $\sigma_1 \neq \sigma_2$. Let ${\cal I}_{\sigma_1,\sigma_2} = \{(t_k, u_k) : k=1,\ldots,[\sigma_1,\sigma_2]\}$ with $t_1 < t_2 < \cdots < t_{[\sigma_1,\sigma_2]}$ be as in \eqref{SharedIndexSet}. For any $q \in \{1,\ldots,Q\}^{[\sigma_1,\sigma_2]}$, define the moment function
\begin{align}
m^{\sigma_1,\sigma_2,q}_{y_0}(y,x,\theta) := \mathbbm{1}\!\left\{(y_{t_1},\ldots,y_{t_{[\sigma_1,\sigma_2]}}) = q\right\} S^{\sigma_1}_{y_0}(y,x,\theta)
- \mathbbm{1}\!\left\{(y_{u_1},\ldots,y_{ u_{[\sigma_1,\sigma_2]}}) = q\right\} S^{\sigma_2}_{y_0}(y,x,\theta).
\label{GeneralizedMoment}
\end{align}
Then
\begin{align}
\mathbb{E}\!\left[m^{\sigma_1,\sigma_2,q}_{Y_0}(Y,X,\theta_0) \,\Big|\, Y_0, X, A\right] = 0.
\label{GeneralizedMomentCondition}
\end{align}
\end{theorem}

When both $\sigma_1$ and $\sigma_2$ are cyclic permutations, the shared index set ${\cal I}_{\sigma_1,\sigma_2}$ reduces to the structure in Theorem~\ref{th:MainTheorem}: the diagonal pairs $(t_k,t_k)$ correspond to the shared initial segment of the generation orders $\tau_1$ and $\tau_2$, and $(s_1,s_2) = (\tau_1(k),\tau_2(k))$ is the first point of disagreement.

\begin{remark}
Analogous to \eqref{MainMomentCondition} in Theorem~\ref{th:MainTheorem}, the result can be written equivalently as: for any $\tilde{q} \in \{1,\ldots,Q\}$,
\begin{align*}
\mathbb{E}\!\left[\mathbbm{1}\{Y_{s_1} = \tilde{q}\} \, S^{\sigma_1}_{Y_0}(Y,X,\theta_0) - \mathbbm{1}\{Y_{s_2} = \tilde{q}\} \, S^{\sigma_2}_{Y_0}(Y,X,\theta_0) \,\Big|\, \{Y_{t_k}\}_{k=1}^{[\sigma_1,\sigma_2]-1}, Y_0, X, A\right] = 0,
\end{align*}
where $\{(t_k,t_k) : k=1,\ldots,[\sigma_1,\sigma_2]-1\}$ are the diagonal elements of ${\cal I}_{\sigma_1,\sigma_2}$ and $(s_1,s_2)$ is the off-diagonal element.
\end{remark}

\subsection{Proof of Theorem~\ref{th:GeneralizedMoments}}

By Lemma~\ref{lemma:SnoAlpha} (which holds for all permutations, not just cyclic ones), we have $S^{\sigma_i}_{y_0}(y,x,\theta) \cdot p(y \mid y_0,x,\theta,\alpha) = p^{\sigma_i}(y \mid y_0,x,\theta,\alpha)$ for $i \in \{1,2\}$. Therefore, establishing \eqref{GeneralizedMomentCondition} is equivalent to showing that for all $\alpha$,
\begin{align}
\sum_{y \in \{1,\ldots,Q\}^T} \left[\mathbbm{1}\!\left\{(y_{t_1},\ldots,y_{t_{[\sigma_1,\sigma_2]}}) = q\right\} p^{\sigma_1}(y) - \mathbbm{1}\!\left\{(y_{u_1},\ldots,y_{u_{[\sigma_1,\sigma_2]}}) = q\right\} p^{\sigma_2}(y)\right] = 0,
\label{ProofTarget}
\end{align}
where we suppress the dependence on $(y_0,x,\theta,\alpha)$ for brevity.

\medskip
\noindent\textbf{Step 1: Factorization of the transformed model over cycles.}
For any permutation $\sigma$ with cycle decomposition $\sigma = c_1 \circ \cdots \circ c_L$, the transformed model \eqref{AuxModelSimplified} can be written as a product over cycles:
\begin{align}
p^\sigma(y) = \prod_{\ell=1}^L \prod_{\substack{t \in \{1,\ldots,T\}:\\ \sigma(t-1) \in C_\ell}} p_t(y_{\sigma(t-1)} \mid y_{t-1}).
\label{CycleFactorization}
\end{align}
To see this, note that $p^\sigma(y) = \prod_{t=1}^T p_t(y_{\sigma(t-1)} \mid y_{t-1})$ by \eqref{AuxModelSimplified}, and the mapping $t \mapsto \sigma(t-1)$ assigns each factor to the cycle containing $\sigma(t-1)$, giving a partition of the product over cycles. We denote the factor corresponding to cycle $C_\ell$ by
\begin{align*}
p^\sigma_{C_\ell}(y) := \prod_{\substack{t \in \{1,\ldots,T\}:\\ \sigma(t-1) \in C_\ell}} p_t(y_{\sigma(t-1)} \mid y_{t-1}).
\end{align*}
For the cycle $C_*$ containing $T$, this factor includes the term at $t=1$ (since $\sigma(0) = \sigma(T) \in C_*$). Note that $p^\sigma_{C_\ell}(y)$ depends on $y$ only through $\{y_s : s \in C_\ell\}$ and $\{y_{t-1} : \sigma(t-1) \in C_\ell\}$. We have $p^\sigma(y) = p^\sigma_{C_*}(y) \cdot \prod_{\ell: C_\ell \neq C_*} p^\sigma_{C_\ell}(y)$.

\medskip
\noindent\textbf{Step 2: Shared cycles contribute identically.}
Since $\sigma_1 \sim \sigma_2$, the non-$T$ cycles are identical. For any shared cycle $C_\ell \neq C_*$, the set $\{t : \sigma_1(t-1) \in C_\ell\}$ and $\{t : \sigma_2(t-1) \in C_\ell\}$ are the same (since $\sigma_1$ and $\sigma_2$ act identically on $C_\ell$). Hence $p^{\sigma_1}_{C_\ell}(y) = p^{\sigma_2}_{C_\ell}(y)$ for all $y$.

Denote the common product of shared-cycle factors by
\begin{align*}
F_{\mathrm{shared}}(y) := \prod_{\ell: C_\ell \neq C_*} p^{\sigma_1}_{C_\ell}(y) = \prod_{\ell: C_\ell \neq C_*} p^{\sigma_2}_{C_\ell}(y).
\end{align*}
Then $p^{\sigma_i}(y) = F_{\mathrm{shared}}(y) \cdot p^{\sigma_i}_{C_*}(y)$ for $i = 1,2$.

\medskip
\noindent\textbf{Step 3: Reduction to the $T$-cycle.}
Write $q = (\bar{q}, \tilde{q})$, where $\bar{q} = (q_1,\ldots,q_{[\sigma_1,\sigma_2]-1})$ corresponds to the diagonal elements of ${\cal I}_{\sigma_1,\sigma_2}$ and $\tilde{q} \in \{1,\ldots,Q\}$ corresponds to the off-diagonal element $(s_1,s_2)$. Then the left-hand side of \eqref{ProofTarget} becomes
\begin{align*}
\sum_{y} \mathbbm{1}\!\left\{y_{t_k}=q_k,\, k=1,\ldots,[\sigma_1,\sigma_2]-1\right\} F_{\mathrm{shared}}(y) \left[\mathbbm{1}\{y_{s_1}=\tilde{q}\} \, p^{\sigma_1}_{C_*}(y) - \mathbbm{1}\{y_{s_2}=\tilde{q}\} \, p^{\sigma_2}_{C_*}(y)\right].
\end{align*}
Now we perform the summation. The indicator $\mathbbm{1}\{y_{t_k}=q_k\}$ fixes the outcomes at positions in $C_{\mathrm{shared}}$ and on the shared initial segment of $C_*$. Since $F_{\mathrm{shared}}(y)$ depends on $y$ only through $\{y_s : s \in C_{\mathrm{shared}}\}$ and certain conditioning variables $y_{t-1}$, and since $p^{\sigma_i}_{C_*}(y)$ depends on $y$ only through $\{y_s : s \in C_*\}$ and certain $y_{t-1}$, we can separate the summation.

The summation over $\{y_s : s \in C_{\mathrm{shared}}\}$ is constrained by the indicators, and yields $F_{\mathrm{shared}}$ evaluated at the fixed values. After this, what remains is a sum over $\{y_s : s \in C_*\}$, with the outcomes on the shared initial segment of $C_*$ also fixed by the indicators:
\begin{align*}
F_{\mathrm{shared}}(\bar{q}) \cdot \sum_{\{y_s : s \in C_*\}} \mathbbm{1}\!\big\{y_{\sigma_1^{j}(T)} = \bar{q}_j,\, j=1,\ldots,k^*\big\} \left[\mathbbm{1}\{y_{s_1}=\tilde{q}\} \, p^{\sigma_1}_{C_*}(y) - \mathbbm{1}\{y_{s_2}=\tilde{q}\} \, p^{\sigma_2}_{C_*}(y)\right],
\end{align*}
where $F_{\mathrm{shared}}(\bar{q})$ absorbs the shared-cycle contributions.

\medskip
\noindent\textbf{Step 4: Telescoping within the $T$-cycle.}
It remains to show that for each $i \in \{1,2\}$,
\begin{align*}
H_i := \sum_{\{y_s : s \in C_*\}} \mathbbm{1}\!\big\{y_{\sigma_i^{j}(T)} = \bar{q}_j,\, j=1,\ldots,k^*\big\} \, \mathbbm{1}\{y_{s_i}=\tilde{q}\} \, p^{\sigma_i}_{C_*}(y)
\end{align*}
is the same for $i=1$ and $i=2$. Within the $T$-cycle, $p^{\sigma_i}_{C_*}(y)$ is a product of transition probabilities along the generation order $\tau_*^{(i)}$. Summing over all $y_s$ for $s \in C_*$ that are \emph{not} fixed by the indicators (i.e., not in the shared initial segment $\{\sigma_i(T),\ldots,\sigma_i^{k^*}(T)\}$ and not the split point $s_i$), the Markov property causes these sums to telescope to one, since $\sum_{r=1}^Q p_t(r \mid \cdot) = 1$.

After the telescoping, we are left with
\begin{align*}
H_i = p_1\!\left(\bar{q}_1 \mid y_0\right) \prod_{j=2}^{k^*} p_{\sigma_i^{j-1}(T)+1}\!\left(\bar{q}_{j} \mid \bar{q}_{j-1}\right) \cdot p_{\sigma_i^{k^*}(T)+1}\!\left(\tilde{q} \mid \bar{q}_{k^*}\right),
\end{align*}
where $\bar{q}_j$ is the value assigned to the shared index $\sigma_1^{j}(T) = \sigma_2^{j}(T)$, and where for $k^*=0$ the expression reduces to $p_1(\tilde{q} \mid y_0)$. Since $\sigma_1^{j}(T) = \sigma_2^{j}(T)$ for $j = 0,1,\ldots,k^*$, every factor is the same for $i=1$ and $i=2$, so $H_1 = H_2$.

Therefore, the expression in \eqref{ProofTarget} equals $F_{\mathrm{shared}}(\bar{q}) \cdot (H_1 - H_2) = 0$, establishing \eqref{GeneralizedMomentCondition}. \hfill $\square$

\bigskip

\noindent\textbf{Example: $T=4$ with non-cyclic permutations.}
For $T=4$, consider $\sigma_1 = [1,2,4][3]$ and $\sigma_2 = [1,4,2][3]$ (in sequence form, $\sigma_1 = (2,4,3,1)$ and $\sigma_2 = (4,1,3,2)$). Both share the cycle $[3]$ (which does not contain $T=4$), and differ in the cycle containing $4$: $[1,2,4]$ versus $[1,4,2]$. Since neither $\sigma_1$ nor $\sigma_2$ is a cyclic permutation of $\{1,2,3,4\}$, the moment conditions from Theorem~\ref{th:GeneralizedMoments} cannot be obtained directly from Theorem~\ref{th:MainTheorem}. The shared index set is ${\cal I}_{\sigma_1,\sigma_2} = \{(3,3),(1,2)\}$, with $[\sigma_1,\sigma_2]=2$, and the corresponding moment conditions are
\begin{align}
\mathbb{E}\!\left[\mathbbm{1}\{Y_1 = q\} \, S^{\sigma_1}_{Y_0}(Y,X,\theta_0) - \mathbbm{1}\{Y_2 = q\} \, S^{\sigma_2}_{Y_0}(Y,X,\theta_0) \,\Big|\, Y_3, Y_0, X, A\right] = 0,
\label{ExampleMoment}
\end{align}
for $q \in \{1,\ldots,Q\}$. We now verify \eqref{ExampleMoment} directly using the chain representation.

\medskip

The cycle $C_* = \{1,2,4\}$ contains $T=4$, and the shared cycle is $\{3\}$. For $\sigma_1$, the generation order within $C_*$ is $\tau_*^{(1)} = (1,2,4)$, meaning the chain visits $y_1$, then $y_2$, then $y_4$. For $\sigma_2$, the generation order is $\tau_*^{(2)} = (2,1,4)$, visiting $y_2$, then $y_1$, then $y_4$. The shared cycle $\{3\}$ is a fixed point, producing a closed-loop factor. Using the chain representation from Section~\ref{app:generalized}.1, the transformed models are
\begin{align*}
p^{\sigma_1}(y) &= \underbrace{p_1(y_1 \mid y_0) \cdot p_2(y_2 \mid y_1) \cdot p_3(y_4 \mid y_2)}_{\text{chain through } C_* = \{1,2,4\}} \cdot \underbrace{p_4(y_3 \mid y_3)}_{\text{loop at } \{3\}}, \\[4pt]
p^{\sigma_2}(y) &= \underbrace{p_1(y_2 \mid y_0) \cdot p_3(y_1 \mid y_2) \cdot p_2(y_4 \mid y_1)}_{\text{chain through } C_* = \{1,2,4\}} \cdot \underbrace{p_4(y_3 \mid y_3)}_{\text{loop at } \{3\}}.
\end{align*}
The closed-loop factor $p_4(y_3 \mid y_3)$ is identical in both expressions, as guaranteed by $\sigma_1 \sim \sigma_2$.

Using $S^{\sigma_i} \cdot p = p^{\sigma_i}$, the moment condition \eqref{ExampleMoment} is equivalent to showing that for all $\alpha$ and for fixed $y_3$,
\begin{align}
\sum_{y_1,y_2,y_4} \mathbbm{1}\{y_1 = q\} \, p^{\sigma_1}(y) = \sum_{y_1,y_2,y_4} \mathbbm{1}\{y_2 = q\} \, p^{\sigma_2}(y).
\label{ExampleTarget}
\end{align}
The common factor $p_4(y_3 \mid y_3)$ cancels from both sides. For the left-hand side, we sum the $C_*$-chain of $\sigma_1$ over the free variables. Sum over $y_4$ first (the last outcome generated by the chain):
\begin{align*}
\sum_{y_4} p_3(y_4 \mid y_2) = 1 \qquad \text{(telescopes)}.
\end{align*}
Then sum over $y_2$:
\begin{align*}
\sum_{y_2} p_2(y_2 \mid y_1) = 1 \qquad \text{(telescopes)}.
\end{align*}
What remains is
$\sum_{y_1} \mathbbm{1}\{y_1 = q\} \cdot p_1(y_1 \mid y_0) = p_1(q \mid y_0)$.

For the right-hand side, we sum the $C_*$-chain of $\sigma_2$ analogously. Sum over $y_4$:
\begin{align*}
\sum_{y_4} p_2(y_4 \mid y_1) = 1 \qquad \text{(telescopes)}.
\end{align*}
Then sum over $y_1$:
\begin{align*}
\sum_{y_1} p_3(y_1 \mid y_2) = 1 \qquad \text{(telescopes)}.
\end{align*}
What remains is
$\sum_{y_2} \mathbbm{1}\{y_2 = q\} \cdot p_1(y_2 \mid y_0) = p_1(q \mid y_0)$.

Both sides equal $p_1(q \mid y_0)$, confirming \eqref{ExampleTarget}. The key mechanism is that both chains start from $y_0$ (via the $p_1$ factor) and visit the same set of outcomes $\{y_1, y_2, y_4\}$, but in different orders. When summing over all the free variables, the intermediate transitions telescope to one, and only the initial transition probability $p_1(q \mid y_0)$ survives --- regardless of the order in which the chain visits the elements.

\section{Details for the general exponential--family, bivariate gamma, and linear--normal examples}
\label{app:moregen-details}

Lemma~\ref{lem:general-suff} states the sufficient statistic for the exponential--family model \eqref{eq:hahn-density} discussed in Section~\ref{sec:moregen}, and Proposition~\ref{prop:general-moment} verifies the validity of the moment condition obtained from Lemma~\ref{lemma:SnoAlphaGeneric}.

\begin{lemma}{\textsc{(Sufficient statistic for the exponential--family model)}.}
\label{lem:general-suff}
Let $(Y_1,\ldots,Y_T)$ be independent with densities $f(y_t,\eta_t,\alpha)$ given by \eqref{eq:hahn-density}, for indices $\eta_1,\ldots,\eta_T$ and common fixed effect $\alpha$. Then $S(Y)=\sum_{t=1}^{T}s(Y_t)$ is sufficient for $\alpha$, and the conditional distribution of $(Y_1,\ldots,Y_T)$ given $S=\bar s$ does not depend on $\alpha$:
\begin{equation}
f(y\mid S=\bar s,x;\theta)=\frac{\exp\!\big(\sum_{t=1}^{T} b(y_t,\eta_t,\theta)\big)}{\displaystyle\int_{\mathcal A_{\bar s}}\exp\!\big(\sum_{t=1}^{T} b(u_t,\eta_t,\theta)\big)\,du},\qquad \mathcal A_{\bar s}=\Big\{u:\textstyle\sum_t s(u_t)=\bar s\Big\}.
\label{eq:general-suff}
\end{equation}
Moreover $S$ is symmetric in its arguments.
\end{lemma}

\begin{proof}
By \eqref{eq:hahn-density}, the joint density factors as
\[
\prod_{t=1}^{T} f(y_t,\eta_t,\alpha)=\exp\!\big(\alpha' e(\theta)\,S(y)\big)\cdot\exp\!\Big(\sum_{t=1}^{T} d(\eta_t,\alpha,\theta)\Big)\cdot\exp\!\Big(\sum_{t=1}^{T} b(y_t,\eta_t,\theta)\Big),
\]
with $S(y)=\sum_t s(y_t)$. The first factor depends on $y$ only through $S(y)$, and the second is free of $y$. Conditioning on $S=\bar s$ divides by the integral of the joint over $\mathcal A_{\bar s}$; the factor $\exp(\alpha' e(\theta)\,\bar s)$, which is constant on $\mathcal A_{\bar s}$, and the $y$--free factor $\exp(\sum_t d(\eta_t,\alpha,\theta))$ cancel between numerator and denominator, leaving \eqref{eq:general-suff}, which is free of $\alpha$. Finally, $S=\sum_t s(Y_t)$ is a sum and is therefore symmetric in $(Y_1,\ldots,Y_T)$.
\end{proof}

\begin{proposition}{\textsc{(Moment conditions for the exponential--family panel model)}.}
\label{prop:general-moment}
Consider the dynamic model \eqref{JointProbGeneric} with period densities \eqref{eq:hahn-density} and $T=3$. Then for any bounded $g$,
\begin{equation}
\mathbb{E}\!\left[\,g(Y_2)\,S^{\sigma_1}_{Y_0}(Y,X,\theta_0)-g(Y_1)\;\Big|\;Y_0,X,A\,\right]=0,
\label{eq:general-moment}
\end{equation}
where $S^{\sigma_1}_{Y_0}$ is given by \eqref{eq:hahn-Sratio} at $\theta=\theta_0$.
\end{proposition}

\begin{proof}
By Lemma~\ref{lemma:SnoAlphaGeneric}, $S^{\sigma_1}_{Y_0}\,p=p^{\sigma_1}$, so
\[
\mathbb{E}\!\left[g(Y_2)\,S^{\sigma_1}_{Y_0}\,\big|\,Y_0,X,A\right]=\int g(y_2)\,p^{\sigma_1}(y\mid Y_0,X,\theta_0,A)\,dy,
\]
where the integral is over the support of $(Y_1,Y_2,Y_3)$. For the reverse--cyclic permutation,
\[
p^{\sigma_1}(y\mid Y_0,X,\theta_0,A)=f(y_2,\eta_1,A)\,f(y_3,\eta_2,A)\,f(y_1,\eta_3,A),
\]
with $\eta_1=z_1(Y_0,X,\theta_0)$ constant given $Y_0$, $\eta_2=z_2(y_1,X,\theta_0)$ depending on $y_1$ but not on $y_3$, and $\eta_3=z_3(y_2,X,\theta_0)$ depending on $y_2$ but not on $y_1$. Integrate sequentially. The factor $f(y_3,\eta_2,A)$ is a density in $y_3$, so $\int f(y_3,\eta_2,A)\,dy_3=1$; after eliminating $y_3$, the factor $f(y_1,\eta_3,A)$ is a density in $y_1$, so $\int f(y_1,\eta_3,A)\,dy_1=1$. The remaining one--dimensional integral is
\[
\int g(y_2)\,f(y_2,\eta_1,A)\,dy_2=\mathbb{E}\!\left[g(Y_1)\,\big|\,Y_0,X,A\right],
\]
because under the true model $Y_1\mid Y_0,X,A$ has density $f(\cdot,\eta_1,A)$ with the same $\eta_1=z_1(Y_0,X,\theta_0)$. Combining the three steps gives \eqref{eq:general-moment}.
\end{proof}

\begin{remark}
The argument is structurally identical to the proof of Lemma~\ref{lemma:cyclic} for the discrete case. The Markov property ensures that $\eta_t$ depends on $y_{t-1}$ but not on $y_t$. As a result, the reverse--cyclic permutation places each $Y_t$ as a free integration variable while its density parameter $\eta_t$ is governed by a different element of the triple, allowing each factor to integrate to unity independently. Summation in the discrete case is replaced by integration, and the indicator $\mathbf{1}\{Y_{\tau(k)}=q\}$ is replaced by the weighting function $g$ evaluated at $Y_{\tau(k)}$.
\end{remark}

The bivariate gamma model \eqref{eq:lindsay-model} illustrates a model whose symmetric sufficient statistic depends on the common parameter, and therefore lies outside Hahn's class. We record the corresponding sufficient statistic and moment condition.

\begin{lemma}{\textsc{(Sufficient statistic for the bivariate gamma model)}.}
\label{lem:lindsay-suff}
Let $\{(Y_{1t},Y_{2t})\}_{t=1}^{T}$ be independent across $t$, with $Y_{1t}\sim\mathrm{Gamma}(z_{1t},A)$ and $Y_{2t}\sim\mathrm{Gamma}(z_{2t},\phi A)$ for given indices $z_{1t},z_{2t}$ in the (shape, rate) parametrization. Then $s_\theta(Y)=\sum_{t=1}^{T} Y_{1t}+\phi\sum_{t=1}^{T} Y_{2t}$ is sufficient for $A$, and it is symmetric in $t$.
\end{lemma}

\begin{proof}
The joint density is
\[
\prod_{t=1}^{T}\phi^{\,z_{2t}}\,\frac{A^{\,z_{1t}+z_{2t}}}{\Gamma(z_{1t})\Gamma(z_{2t})}\,y_{1t}^{\,z_{1t}-1}y_{2t}^{\,z_{2t}-1}\,e^{-A(y_{1t}+\phi y_{2t})}=e^{-A\,s_\theta(y)}\,A^{\sum_t(z_{1t}+z_{2t})}\,h(y),
\]
with $s_\theta(y)=\sum_t y_{1t}+\phi\sum_t y_{2t}$ and $h(y)=\prod_t \phi^{\,z_{2t}}\Gamma(z_{1t})^{-1}\Gamma(z_{2t})^{-1}y_{1t}^{\,z_{1t}-1}y_{2t}^{\,z_{2t}-1}$ free of $A$. The first two factors depend on $y$ only through $s_\theta(y)$, so by the factorization theorem $s_\theta(Y)$ is sufficient for $A$; being a sum over $t$, it is symmetric. Unlike the canonical statistic in \eqref{eq:hahn-density}, $s_\theta$ depends on $\theta$ through $\phi$.
\end{proof}

By Lemma~\ref{lemma:SnoAlphaGeneric}, $S^{\sigma}_{y_0}=p^{\sigma}/p$ is free of $A$. In the ratio, the factors $\phi^{z_{2t}}$, $A^{z_{1t}+z_{2t}}$ and the gamma functions cancel period by period; the $A$--coupling cancels because $\sum_t y_{1,\sigma(t-1)}=\sum_t y_{1t}$ and $\sum_t y_{2,\sigma(t-1)}=\sum_t y_{2t}$, so that $\phi$ disappears together with $A$; and the exponents $-1$ on $y_{1t}$ and $y_{2t}$ cancel by the same symmetry. What remains is \eqref{eq:lindsay-Sratio}, which depends on $\theta$ only through $(\beta_1,\beta_2)$.

\begin{proposition}{\textsc{(Moment conditions for the bivariate gamma model)}.}
\label{prop:lindsay-moment}
For the dynamic model \eqref{eq:lindsay-model} with $T=3$, the moment condition \eqref{eq:lindsay-moment} holds for every bounded $g$, with $S^{\sigma_1}_{Y_0}$ given by \eqref{eq:lindsay-Sratio}.
\end{proposition}

\begin{proof}
The argument is identical to the proof of Proposition~\ref{prop:general-moment}, with the bivariate gamma transition density $f(y_t,z_t,A)$ (the product of the two gamma densities) in place of the period density. Under the reverse--cyclic permutation, $p^{\sigma_1}=f(y_2,z_1,A)\,f(y_3,z_2,A)\,f(y_1,z_3,A)$, with $z_1=z(Y_0)$ constant given $Y_0$, $z_2=z(y_1)$ depending on $y_1$ but not $y_3$, and $z_3=z(y_2)$ depending on $y_2$ but not $y_1$. Integrating out $y_3=(y_{13},y_{23})$ and then $y_1=(y_{11},y_{21})$ each contributes unity, and the remaining integral equals $\mathbb{E}[g(Y_1)\mid Y_0,X,A]$, since under the true model $Y_1\mid Y_0,X,A$ has density $f(\cdot,z_1,A)$.
\end{proof}

\subsection{Moments of the linear--normal moment function}
\label{app:moregen-details:linear}
The discussion of the linear panel data model in Section~\ref{sec:moregen} noted that the moment function \eqref{eq:linear-Sratio} is the exponential of a quadratic form in the Gaussian outcomes. It is therefore heavy--tailed, with conditional moments that exist only when the state--dependence coefficient lies in a bounded interval. This appendix section makes the claim in Footnote {\ref{Footnote: Existence of moments}} precise. The results build on Chapter 3 of \citet{mathai1992quadratic} which establishes that if $W\sim N(\mu,\Sigma)$ in $\mathbb{R}^n$,  $C$ is symmetric, and  $b\in \mathbb{R}^n$ then
\begin{equation}\label{eq:master}
\;
E\!\left[\exp\!\big(W'CW+b'W\big)^{k}\right]
=\det(I-2k\Sigma C)^{-1/2}\,
\exp\!\Big(\tfrac12\,v_k'P_k^{-1}v_k-\tfrac12\,\mu'\Sigma^{-1}\mu\Big),
\;
\end{equation}
where
\[
P_k=\Sigma^{-1}-2kC,\qquad v_k=kb+\Sigma^{-1}\mu .
\]
 The $k$-th moment is finite if and only if $\Sigma^{-1}-2kC$ is positive definite.

Take the linear index $z_t(y_{t-1},x,\theta)=\gamma\,y_{t-1}+x_t'\beta$, so that \eqref{eq:linear-model} is the Gaussian first--order autoregressive model with a fixed effect, and let $T=3$ with $\sigma_1$ the reverse cycle. Writing $S^{\sigma_1}_{y_0}=\exp\!\big(Q(y)/\sigma^2\big)$, expression \eqref{eq:linear-Sratio} gives
\begin{equation}
Q(y)=z_1(y_2-y_1)+z_2(y_3-y_2)+z_3(y_1-y_3),\quad \text{with} \quad z_t=\gamma y_{t-1}+x_t'\beta.
\label{eq:linear-Q}
\end{equation}
Expanding and collecting terms, the cross terms in $y_1y_2$ cancel, leaving
\begin{equation}
Q(y)=\underbrace{\gamma\,y_3(y_1-y_2)}_{\text{quadratic}}\;+\;\underbrace{\gamma\,y_0(y_2-y_1)+\beta'\big[x_1(y_2-y_1)+x_2(y_3-y_2)+x_3(y_1-y_3)\big]}_{\text{linear in }y}.
\label{eq:linear-Qdecomp}
\end{equation}

Conditional on $(Y_0,X,A)$, the recursion $Y_t=\gamma Y_{t-1}+x_t'\beta+A+\varepsilon_t$
reads $JY=\mu_0+\varepsilon$ with $\varepsilon\sim\mathcal{N}(0,\sigma^2 I_3)$, so $Y$ is
normal with finite mean and variance $\sigma^2P^{-1}$, where
\[
P=\begin{pmatrix}1+\gamma^2&-\gamma&0\\-\gamma&1+\gamma^2&-\gamma\\0&-\gamma&1\end{pmatrix}.
\]
Writing $M$ for the Hessian of $Q$, so that its quadratic part is $\tfrac12\,y'My$,
\[
M=\begin{pmatrix}0&0&\gamma\\0&0&-\gamma\\\gamma&-\gamma&0\end{pmatrix},
\]
and noting that $Q$ has no term that is constant in $y$, we can write
$\big(S^{\sigma_1}_{Y_0}\big)^k=\exp\!\big(Y'\tfrac{1}{2\sigma^2}M\,Y+b'Y\big)^{k}$, with
$b\in\mathbb{R}^3$ collecting the linear coefficients of $Q/\sigma^2$. This is exactly the
form to which the result in \eqref{eq:master} applies, with $\Sigma=\sigma^2P^{-1}$ and $C=\tfrac{1}{2\sigma^2}M$. Equation 
\eqref{eq:master} then tells us that the moment is finite when
$
\Sigma^{-1}-2k\cdot\tfrac{1}{2\sigma^2}M=\sigma^{-2}\,(P-kM)
$
is positive definite, i.e.\ when $P-kM$ is positive definite. 
Finiteness is therefore governed by $\gamma$ alone, and not by $\sigma^2$, $A$, $Y_0$, or $X$.

It remains to turn positive definiteness of $P-kM$ into a condition on $\gamma$. The two leading principal minors
of $P-kM$ are $1+\gamma^2$ and $1+\gamma^2+\gamma^4$. So positive definiteness comes down to
the sign of the full determinant. A direct computation gives
\[
\det(P-kM)=1-2k(k-1)\,\gamma^2\big(1-\gamma+\gamma^2\big),
\]
so $\mathbb{E}\big[(S^{\sigma_1}_{Y_0})^k\mid Y_0,X,A\big]<\infty$ if and only if
$2k(k-1)\,\gamma^2(1-\gamma+\gamma^2)<1$.

Because $1-\gamma+\gamma^2=(\gamma-\tfrac12)^2+\tfrac34>0$, the region for $\gamma$ for which the moment is finite is an interval of $\gamma$ around the origin that shrinks as $k$ grows. For $k=1$ the condition is satisfied for all $\gamma$; this is consistent with $\mathbb{E}[S^{\sigma_1}_{Y_0}\mid Y_0,X,A]=1$. The second moment ($k=2$) requires $\gamma$ to lie between the two real roots of $1-4\gamma^2+4\gamma^3-4\gamma^4$, that is $\gamma\in(-0.400,\,0.575)$, and the higher moments require the successively smaller intervals
\[
k=3:\ \ \gamma\in({-0.251},\,{0.326}),\qquad\qquad k=4:\ \ \gamma\in({-0.184},\,{0.224}).
\]

\section{Proof of Proposition \ref{prop:Identification_Habit}}
\label{app:identification}
Theorem \ref{thm_GN}, stated in Subsection \ref{app:identification:auxilary_result}, restates Theorem 4 of \citet{gale1965jacobian} and is included for convenience. This result is the key tool used to establish point identification of $\theta_0$ for our empirical specification \eqref{EmpiricalSpec} in Proposition \ref{prop:Identification_Habit}. Lemma \ref{lemma:pmatrix} complements the theorem by providing simple sufficient conditions for a matrix to be a $P$-matrix, the main condition required to apply Theorem \ref{thm_GN}. Subsection \ref{app:identification:proof} then gives a detailed proof of Proposition \ref{prop:Identification_Habit}.

\subsection{Auxiliary results} \label{app:identification:auxilary_result}

\begin{theorem}{\textsc{(Theorem 4 in \cite{gale1965jacobian})}.} \label{thm_GN}
    If $\Xi:\Theta\mapsto \mathbb{R}^{Q+1}$, where $\Theta$ is a closed rectangular region of $\mathbb{R}^{Q+1}$, is a differentiable mapping such that the jacobian matrix $\nabla \Xi(\theta)$ is a P-matrix for all $\theta\in \Theta$, then $\Xi$ is univalent in $\Theta$.
\end{theorem}

\begin{lemma} \label{lemma:pmatrix}
Let $J\in\mathbb{R}^{n\times n}$. Suppose that (i) $J_{ij}>0$ for all $i,j\in\{1,\dots,n\}$, and (ii) that $J$ is strictly diagonally dominant, i.e.
\begin{align*}
    \forall i \in \{1,\ldots,n\}: \quad J_{ii}>\sum_{j\neq i} J_{ij}
\end{align*}
Then $J$ is a $P$-matrix (that is, has all principal minors positive).
\end{lemma}

\begin{proof}
    If $J\in\mathbb{R}^{n\times n}$ satisfies (i) and (ii), then it is clear that any principal submatrix of $J$ also satisfies (i) and (ii). Therefore, it suffices to show that (i) and (ii) imply that $\text{det}(J)>0$. \\

   \noindent  To this end, it is useful to write $J = D+R$ where $D=\text{diag}(J_{11},\ldots,J_{nn})$. Since $\text{det}(D)=\prod\limits_{i=1}^n J_{ii}>0$ by (i), $D$ is invertible and we can further write $J=D(I+D^{-1}R)$. Since $\text{det}(J) = \text{det}(D)\times \text{det}(I+D^{-1}R)$, the sign of $\text{det}(J)$ is determined by that of $\text{det}(I+D^{-1}R)$. By (ii), $D^{-1}R$ has row-sums strictly less than 1. Hence, by the Gershgorin disc theorem (e.g Theorem 6.1.1 \cite{horn2012matrix}), all eigenvalues of $D^{-1}R$ have moduli strictly less than one. It follows that any real eigenvalue of $(I+D^{-1}R)$ is strictly positive and thus contributes positively to $\text{det}(I+D^{-1}R)$. Furthermore,  since $(I+D^{-1}R)\in \mathbb{R}^{n\times n}$, any of its complex eigenvalues admits a conjugate pair, and that pair contributes positively to $\text{det}(I+D^{-1}R)$. This establishes that $\text{det}(I+D^{-1}R)>0$ which concludes the proof.
\end{proof}

\subsection{Proof details}  \label{app:identification:proof}

\paragraph{Preliminary calculations.} 
It is helpful to express \eqref{moment_exp} as
 \begin{align*} 
   &m_{y_0}^{\sigma_{1},q}(y,x,\theta) \\
&=\mathbbm{1}\{y_1=q,y_2=q,y_3=q\}S^{\sigma_{1}}_{y_0}(y,x,\theta) \\
&+\sum_{\substack{r=1 \\ r \neq q}}^Q\mathbbm{1}\{y_1=r,y_2=q,y_3=r\}S^{\sigma_{1}}_{y_0}(y,x,\theta) \\
&+\sum_{\substack{r=1,\ l=1\\ r\neq q,\ l\neq q,\ r\neq l}}^Q\mathbbm{1}\{y_1=r,y_2=q,y_3=l\}S^{\sigma_{1}}_{y_0}(y,x,\theta) \\
&+\sum_{\substack{r=1 \\ r \neq q}}^Q\mathbbm{1}\{y_1=r,y_2=q,y_3=q\}S^{\sigma_{1}}_{y_0}(y,x,\theta) \\
&+\sum_{\substack{r=1 \\ r \neq q}}^Q\mathbbm{1}\{y_1=q,y_2=q,y_3=r\}S^{\sigma_{1}}_{y_0}(y,x,\theta) \\
&- \mathbbm{1}\big\{ y_1 = q \big\}
\end{align*}
and to note that, for $y_2=q$,
\begin{align*}
& S^{\sigma_{1}}_{y_0}(y,x,\theta) \\
&= \exp(z_{3}(y_1,q,x,\theta)-z_{1}(y_1,y_0,x,\theta)+z_{1}(q,y_0,x,\theta)-z_{2}(q,y_1,x,\theta)+z_{2}(y_3,y_1,x,\theta)-z_{3}(y_3,q,x,\theta)) \\
&=\exp(\left(x_{31,y_{1}}+x_{12,q}+x_{23,y_3}\right)\beta + \gamma_{q}\left(\mathbbm{1}\{y_0=q\}-\mathbbm{1}\{y_3=q\}\right)-\gamma_{y_0}\mathbbm{1}\{y_1=y_0\}+\gamma_{y_1}\mathbbm{1}\{y_3=y_1\})
\end{align*}

Thus, for the initial condition $y_0=q$, we get
\begin{align} \label{disaggregated_moment}
    m_{q}^{\sigma_{1},q}(y,x,\theta)&=\mathbbm{1}\{y_1=q,y_2=q,y_3=q\}- \mathbbm{1}\big\{ y_1 = q \big\} \\
   &+\sum_{\substack{r=1 \\ r \neq q}}^Q\mathbbm{1}\{y_1=r,y_2=q,y_3=r\}\exp((x_{12,q}+x_{21,r})\beta+\gamma_{q}+\gamma_{r}) \notag \\
    &+\sum_{\substack{r=1,\ l=1\\ r\neq q,\ l\neq q,\ r\neq l}}^Q\mathbbm{1}\{y_1=r,y_2=q,y_3=l\}\exp((x_{31,r}+x_{12,q}+x_{23,l})\beta+\gamma_{q})  \notag\\
    &+\sum_{\substack{r=1 \\ r \neq q}}^Q\mathbbm{1}\{y_1=r,y_2=q,y_3=q\}\exp((x_{13,q}+x_{31,r})\beta)  \notag\\
    &+\sum_{\substack{r=1 \\ r \neq q}}^Q\mathbbm{1}\{y_1=q,y_2=q,y_3=r\}\exp((x_{32,q}+x_{23,r})\beta) \notag
\end{align}
As a result, 
\begin{align} \label{gradient_beta}
    &\pdv{m_{q}^{\sigma_{1},q}(y,x,\theta)}{\beta}\\
    &=\sum_{\substack{r=1 \\ r \neq q}}^Q\mathbbm{1}\{y_1=r,y_2=q,y_3=r\}(x_{12,q}+x_{21,r})\exp((x_{12,q}+x_{21,r})\beta+\gamma_{q}+\gamma_{r}) \notag\\
    &+\sum_{\substack{r=1,\ l=1\\ r\neq q,\ l\neq q,\ r\neq l}}^Q\mathbbm{1}\{y_1=r,y_2=q,y_3=l\}(x_{31,r}+x_{12,q}+x_{23,l})\exp((x_{31,r}+x_{12,q}+x_{23,l})\beta+\gamma_{q}) \notag\\
    &+\sum_{\substack{r=1 \\ r \neq q}}^Q\mathbbm{1}\{y_1=r,y_2=q,y_3=q\}(x_{13,q}+x_{31,r})\exp((x_{13,q}+x_{31,r})\beta) \notag\\
    &+\sum_{\substack{r=1 \\ r \neq q}}^Q\mathbbm{1}\{y_1=q,y_2=q,y_3=r\}(x_{32,q}+x_{23,r})\exp((x_{32,q}+x_{23,r})\beta) \notag
\end{align}

\begin{align} \label{gradient_gamma_q}
    \pdv{m_{q}^{\sigma_{1},q}(y,x,\theta)}{\gamma_q}&=\sum_{\substack{r=1 \\ r \neq q}}^Q\mathbbm{1}\{y_1=r,y_2=q,y_3=r\}\exp((x_{12,q}+x_{21,r})\beta+\gamma_{q}+\gamma_{r}) \\
    &+\sum_{\substack{r=1,\ l=1\\ r\neq q,\ l\neq q,\ r\neq l}}^Q\mathbbm{1}\{y_1=r,y_2=q,y_3=l\}\exp((x_{31,r}+x_{12,q}+x_{23,l})\beta+\gamma_{q}) \notag
\end{align}
and for $r\neq q$
\begin{align} \label{gradient_gamma_r}
    \pdv{m_{q}^{\sigma_{1},q}(y,x,\theta)}{\gamma_r}&=\mathbbm{1}\{y_1=r,y_2=q,y_3=r\}\exp((x_{12,q}+x_{21,r})\beta+\gamma_{q}+\gamma_{r}) 
\end{align}

\paragraph{Constructing a P-matrix.}  Let $\Xi^{1,\beta}(\theta) = \mathbb{E}\left[m_{y_{0}}^{\sigma_{1},1}(Y,X,\theta)|Y_{0}=1,X\in \mathcal{X}^{\beta}_{1}\right]$ and for $q=1,\ldots,Q$, let $\Xi_{\epsilon_{q}}^{q,\gamma}(\theta)=\mathbb{E}\left[m_{y_{0}}^{\sigma_{1},q}(Y,X,\theta)|Y_{0}=q,X\in \mathcal{X}^{\gamma}_{q,\epsilon_{q}}\right]$. Let 
    \begin{align*}
    \Xi(\theta) = \begin{pmatrix}
      \Xi^{1,\beta}(\theta) \\
\Xi_{\epsilon_{1}}^{1,\gamma}(\theta) \\
\vdots \\
\Xi_{\epsilon_{Q}}^{Q,\gamma}(\theta)
    \end{pmatrix}
\end{align*}
We want to show that $\forall \theta \in \Theta$, $\nabla \Xi(\theta)$ is a $P$-matrix for an appropriate selection of positive constants $\epsilon_{1},\ldots,\epsilon_{Q}$. The first observation is that  Assumptions \ref{ass:Compact_Theta} and \ref{ass:Identification_Restrictions}, conditions \eqref{ass:Well_Defined_Sets_beta} and \eqref{ass:Well_Defined_Sets_gamma}, together with the definitions of the sets $\mathcal{X}^{\gamma}_{q,\epsilon}$ and $\mathcal{X}^{\beta}_{q}$ ensure that all entries of $\nabla \Xi(\theta)$ are strictly positive. This is easily seen by inspecting equations \eqref{gradient_beta},\eqref{gradient_gamma_q},\eqref{gradient_gamma_r}. Hence by Lemma \ref{lemma:pmatrix}, all we need to show is that $\nabla \Xi(\theta)$ is strictly row diagonally dominant. \\

\noindent From \eqref{gradient_beta},\eqref{gradient_gamma_q},\eqref{gradient_gamma_r}, we have
\begin{align*}
     &\pdv{m_{q}^{\sigma_{1},q}(y,x,\theta)}{\beta}-\pdv{m_{q}^{\sigma_{1},q}(y,x,\theta)}{\gamma_q}-\sum_{r\neq q}\pdv{m_{q}^{\sigma_{1},q}(y,x,\theta)}{\gamma_r} \\
     &  =+\sum_{\substack{r=1 \\ r \neq q}}^Q\mathbbm{1}\{y_1=r,y_2=q,y_3=r\}(x_{12,q}+x_{21,r}-2)\exp((x_{12,q}+x_{21,r})\beta+\gamma_{q}+\gamma_{r}) \\
    &+\sum_{\substack{r=1,\ l=1\\ r\neq q,\ l\neq q,\ r\neq l}}^Q\mathbbm{1}\{y_1=r,y_2=q,y_3=l\}(x_{31,r}+x_{12,q}+x_{23,l}-1)\exp((x_{31,r}+x_{12,q}+x_{23,l})\beta+\gamma_{q}) \\
    &+\sum_{\substack{r=1 \\ r \neq q}}^Q\mathbbm{1}\{y_1=r,y_2=q,y_3=q\}(x_{13,q}+x_{31,r})\exp((x_{13,q}+x_{31,r})\beta) \\
    &+\sum_{\substack{r=1 \\ r \neq q}}^Q\mathbbm{1}\{y_1=q,y_2=q,y_3=r\}(x_{32,q}+x_{23,r})\exp((x_{32,q}+x_{23,r})\beta)
\end{align*}
For any $x\in \mathcal{X}^{\beta}_{q}$, $x_{12,q}+x_{21,r}-2>0$ and $(x_{31,r}+x_{12,q}+x_{23,l}-1)>0$ in addition to $(x_{13,q}+x_{31,r})>0$ and $(x_{32,q}+x_{23,r})>0$. It readily follows, choosing $q=1$, that 
\begin{align} \label{diagodominance_beta}
    \nabla_{\beta}\Xi^{1,\beta}(\theta)>\sum_{l=1}^Q \nabla_{\gamma_{l}}\Xi^{1,\beta}(\theta)
\end{align}
which means that the first row of $\nabla \Xi(\theta)$ satisfies the condition of strict diagonal dominance.
\\

\noindent Next, we would like to show that for all $q\in \{1,\ldots,Q\}$: 
\begin{align} \label{diagodominance_gamma}
    \nabla_{\gamma_{q}}\Xi_{\epsilon_{q}}^{q,\gamma}(\theta)>\nabla_{\beta}\Xi_{\epsilon_{q}}^{q,\gamma}(\theta)+\sum_{r\neq q}\nabla_{\gamma_{r}}\Xi_{\epsilon_{q}}^{q,\gamma}(\theta)
\end{align}
for appropriate constants $\epsilon_{q}>0$. Let
\begin{align*}
    &\Delta_{q,\epsilon_{q}}(\theta)\\
    &=\nabla_{\gamma_{q}}\Xi_{\epsilon_{q}}^{q,\gamma}(\theta)-\sum_{r\neq q}\nabla_{\gamma_{r}}\Xi_{\epsilon_{q}}^{q,\gamma}(\theta) \\
    &=\mathbb{E}\left[\sum_{\substack{r=1,\ l=1\\ r\neq q,\ l\neq q,\ r\neq l}}^QP(Y_1=r,Y_2=q,Y_3=l|Y_{0}=q,X)\exp((X_{31,r}+X_{12,q}+X_{23,l})\beta+\gamma_{q})|Y_{0}=q,X\in \mathcal{X}^{\gamma}_{q,\epsilon_{q}}\right] \end{align*}
and 
\begin{align*}
    \nu_{q,\epsilon_{q}}(\theta) &= \mathbb{E}\Big[\sum_{\substack{r=1 \\ r \neq q}}^QP(Y_1=r,Y_2=q,Y_3=r|Y_{0}=q,X)\exp((X_{12,q}+X_{21,r})\beta+\gamma_{q}+\gamma_{r}) \\
    &+\sum_{\substack{r=1 \\ r \neq q}}^QP(Y_1=r,Y_2=q,Y_3=q|Y_{0}=q,X)\exp((X_{13,q}+X_{31,r})\beta) \\
    &+\sum_{\substack{r=1 \\ r \neq q}}^QP(Y_1=q,Y_2=q,Y_3=r|Y_{0}=q,X)\exp((X_{32,q}+X_{23,r})\beta)|Y_{0}=q,X\in \mathcal{X}^{\gamma}_{q,\epsilon_{q}}\Big]
\end{align*}

\noindent By \eqref{gradient_beta} and the definition of  $\mathcal{X}^{\gamma}_{q,\epsilon_{q}}$, we have $\nabla_{\beta}\Xi_{\epsilon_{q}}^{q,\gamma}(\theta)< 2\epsilon_{q}\nu_{q,\epsilon_{q}}(\theta)+3\epsilon_{q}\Delta_{q,\epsilon_{q}}(\theta)$. Strict row diagonal dominance will then follow if we can justify the existence of $\epsilon_{q}$ such that $\nabla_{\beta}\Xi_{\epsilon_{q}}^{q,\gamma}(\theta)< 2\epsilon_{q}\nu_{q,\epsilon_{q}}(\theta)+3\epsilon_{q}\Delta_{q,\epsilon_{q}}(\theta) < \Delta_{q,\epsilon_{q}}(\theta)$, i.e that 
\begin{align} \label{sufficient_condition_diagodominance}
    2\epsilon_{q}\nu_{q,\epsilon_{q}}(\theta)<(1-3\epsilon_{q})\Delta_{q,\epsilon_{q}}(\theta)
\end{align}
To this end, we note the upper bound
\begin{align} \label{th:Identification_Habit_bound1}
   2\epsilon_{q}\nu_{q,\epsilon_{q}}(\theta)\leq 2\epsilon_{q}(Q-1) e^{2\epsilon_q \max\{|\underline{\beta}|,|\overline{\beta}|\}} (e^{2\overline{\gamma}}+2)
\end{align}
and the lower bound
\begin{align*}
\Delta_{q,\epsilon_{q}}(\theta) 
&\geq  
P\!\Big(
Y_1\neq q,\ Y_2=q,\ Y_3\neq q,\ Y_1\neq Y_3
\ \Big|\ 
Y_{0}=q,X\in \mathcal{X}^{\gamma}_{q,\epsilon_{q}}
\Big) \exp\Big(3\epsilon_{q}\underline{\beta}+\underline{\gamma}\Big)
\end{align*}
We now proceed to show that $P\!\Big(
Y_1\neq q,\ Y_2=q,\ Y_3\neq q,\ Y_1\neq Y_3
\ \Big|\ 
Y_{0}=q,X\in \mathcal{X}^{\gamma}_{q,\epsilon_{q}}
\Big)$ is uniformly bounded away from zero.  First, remark that for $x\in \mathcal{X}^{\gamma}_{q,\epsilon_{q}}$, 
\begin{align*}
&P\!\Big(
Y_1\neq q,\ Y_2=q,\ Y_3\neq q,\ Y_1\neq Y_3
\ \Big|\ 
Y_{0}=q,X=x,A=\alpha
\Big) \\
&\geq P\!\Big(
Y_1= r,\ Y_2=q,\ Y_3=l
\ \Big|\ 
Y_{0}=q,X=x,A=\alpha
\Big)
\end{align*}
for any $r\neq l, r\neq q, l\neq q$. Fix such a pair $(r,l)$. This is possible because we assumed $Q\geq 3$.  Now, since $x\in \mathcal{X}^{\gamma}_{q,\epsilon_{q}}$, $\norm{x}_{\infty}< B$. Furthermore, $\Theta$ is compact by Assumption \ref{ass:Compact_Theta}. Together with Assumption \ref{ass:Identification_Restrictions}, the multinomial logit specification \eqref{JointProb}, and \eqref{EmpiricalSpec}, this implies that on the compact region $\mathcal{A}^{reg}_{q}$, there exists a constant $\kappa_{q}>0$ such that 
\begin{align*}
    P\!\Big(
Y_1= r,\ Y_2=q,\ Y_3=l
\ \Big|\ 
Y_{0}=q,X=x,A=\alpha
\Big)\geq\kappa_{q}
\end{align*}
By Assumption \ref{ass:Identification_Restrictions}, there also exists $\rho_{y_{0}}>0$ such that for all $x\in \mathcal{X}_{y_0}^{reg}$, $P(A\in \mathcal{A}^{reg}_{y_{0}} |Y_{0}=y_{0},X=x)\geq \rho_{y_{0}}$. Averaging over $\alpha$, we then get that  $\forall x\in \mathcal{X}^{\gamma}_{q,\epsilon_{q}}\subset \mathcal{X}_{q}^{reg}$,
\begin{align*}
P\!\Big(Y_1= r,\ Y_2=q,\ Y_3=l\ \Big|\ Y_{0}=q,X=x\Big) \geq \kappa_{q}\rho_{q}
\end{align*}
where $\kappa_{q}\rho_{q}$ does not depend on $\epsilon_{q}$. Finally, by condition \eqref{ass:Well_Defined_Sets_gamma}, $P(Y_{0}=q,X\in  \mathcal{X}^{\gamma}_{q,\epsilon})>0$ so
\begin{align*}
    &P\!\Big(Y_1\neq q,\ Y_2=q,\ Y_3\neq q,\ Y_1\neq Y_3\ \Big|\ Y_{0}=q,X\in \mathcal{X}^{\gamma}_{q,\epsilon} \Big)\\
    &\geq P\!\Big(Y_1= r,\ Y_2=q,\ Y_3=l\ \Big|\ Y_{0}=q,X\in \mathcal{X}^{\gamma}_{q,\epsilon}\Big) \\
    &= \mathbb{E}\left[ P\!\Big(Y_1= r,\ Y_2=q,\ Y_3=l\ \Big|\ Y_{0}=q,X=x\Big)|Y_{0}=q,X\in \mathcal{X}^{\gamma}_{q,\epsilon}\right]\\
    &\geq \kappa_{q}\rho_{q}>0
\end{align*}
which in turn implies 
\begin{align} \label{th:Identification_Habit_bound2}
\Delta_{q,\epsilon_{q}}(\theta) 
\geq  \kappa_{q}\rho_{q} \exp\Big(3\epsilon_{q}\underline{\beta}+\underline{\gamma}\Big)    
\end{align}
We now observe that $\lim\limits_{\epsilon_q \to 0} 2\epsilon_q \nu_{q,\epsilon_q}(\theta) = 0$ by \eqref{th:Identification_Habit_bound1} while $\lim\limits_{\epsilon_q \to 0} (1-3\epsilon_{q})\Delta_{q,\epsilon_{q}}(\theta) \geq  \kappa_{q}\rho_{q}\exp(\underline{\gamma})>0$ by \eqref{th:Identification_Habit_bound2}. Hence, it is clear that there exists $\epsilon_{q}>0$ sufficiently small such that \eqref{sufficient_condition_diagodominance} is satisfied for all $\theta\in \Theta$.
\\

Taking stock, we have shown that we can select constants $\epsilon_{1},\ldots,\epsilon_{Q}>0$ such that \eqref{diagodominance_gamma} holds for every $q=1,\ldots,Q$. Henceforth, we fix these appropriate constants. Together with \eqref{diagodominance_beta}, this shows that $\nabla \Xi(\theta)$ is strictly row diagonally dominant with positive entries for any $\theta\in \Theta$. Therefore, $\nabla \Xi(\theta)$ is a $P$-matrix by Lemma \ref{lemma:pmatrix}, and Theorem \ref{thm_GN} guarantees that $\Xi(\theta)$ is univalent. To conclude, we note that by Theorem \ref{th:MainTheorem}, $\Xi(\theta_{0})=0$. Consequently,  the univalence of $\Xi$ ensures that $\theta_{0}$ is the unique root of the $(Q+1)$-dimensional system of equations $\Xi(\theta)=0$.

\end{document}

%% file: Tables/table_Nielsen.tex
\begin{tabular}{lccccc} \toprule  \toprule  
  & \multicolumn{2}{c}{MLE without heterogeneity} & & GMM Allowing for Fixed Effects \\ 
\cline{2-3} \cline{5-5} 
  & $\geq 2$ purchases & $\geq 4$ purchases & & $\geq 4$ purchases \\ 
\\ 
Cola ($\gamma$) & 1.525 & 1.523 & & 0.140 & \\ 
& (0.050) & (0.054) & & (0.063) &   \\ 
Root/Ginger ($\gamma$) & 1.564 & 1.611 & & 0.164 & \\ 
& (0.065) & (0.071) & & (0.076) &   \\ 
Citrus ($\gamma$) & 1.514 & 1.542 &  &0.017 & \\ 
& (0.070) & (0.074) & & (0.081) &   \\ 
Unique Flavor ($\gamma$) & 2.493 & 2.520 &  &0.293 & \\ 
& (0.101) & (0.106) & & (0.110) &   \\ 
Other ($\gamma$) & 1.596 & 1.644 &  &0.081 & \\ 
& (0.055) & (0.059) & & (0.068) &   \\ 
Price ($\beta$) & $-0.110$ & $-0.115$ &  &$-0.110$ & \\ 
& (0.017) & (0.019) & & (0.025) &   \\ 
\\ 
Observations & 39,833 & 36,859 &  &36,859 & \\ 
\bottomrule \bottomrule 
\end{tabular}